\documentclass[preprint,final]{elsarticle}

\usepackage[T1]{fontenc} % Use 8-bit encoding that has 256 glyphs
\usepackage{microtype} % Slightly tweak font spacing for aesthetics

\usepackage{amssymb, amsmath, amsthm}

\numberwithin{equation}{section}

\usepackage{graphicx}

\usepackage[hmarginratio=1:1,top=32mm,columnsep=20pt]{geometry} % Document margins
\usepackage{multicol} % Used for the two-column layout of the document
\usepackage[hang, small,labelfont=bf,up,textfont=it,up]{caption} % Custom captions under/above
\usepackage{booktabs} % Horizontal rules in tables
\usepackage{float} % Required for tables and figures in the multi-column environment - they need to
\usepackage[colorlinks=true, urlcolor=blue, pdfborder={0 0 0}]{hyperref} % For hyperlinks in the PDF

\usepackage{paralist} % Used for the compactitem environment which makes bullet points with less space
\usepackage[algoruled]{algorithm2e} % For algorithms 
\usepackage{enumitem}

\usepackage{subfig}

\usepackage{setspace, ifdraft, blindtext} % Spacing in draft versions of papers

\ifdraft{\doublespace}{\singlespace} 
\ifdraft{\usepackage{endfloat}}{}

\usepackage{changes}

\makeatletter
\@namedef{Changes@AuthorColor}{red}
\colorlet{Changes@Color}{red}
\makeatother

\usepackage[numbers]{natbib}

\setcitestyle{authoryear, open={(}, close={)} }

\renewcommand{\vec}[1]{\mathbf{#1}}

\def\ie{\rm{{\em i.e.}\ }}
\def\eg{\rm{{\em e.g.}\ }}

\newtheorem{theorem}{Theorem}

\newtheorem{proposition}[theorem]{Proposition}

\theoremstyle{definition}

\theoremstyle{remark}

\journal{Computational Statistics \& Data Analysis}
\begin{document}

\ifdraft{\listofchanges[style=list]}{}

\clearpage
\setcounter{page}{1}
\begin{frontmatter}

\title{{\small{\color{red} A. Garbuno-Inigo, F.A. DiazDelaO, K.M. Zuev, Gaussian
process hyper-parameter estimation using Parallel Asymptotically Independent
Markov Sampling, Computational Statistics \& Data Analysis, Volume 103, November
2016.}} \\ Gaussian process hyper-parameter estimation using parallel
asymptotically independent Markov sampling } % Article title

\author[liv]{A. Garbuno-Inigo \corref{cor1}}
\ead{agarbuno@liv.ac.uk}

\author[liv]{F. A. DiazDelaO }

\author[liv]{K. M. Zuev}

\cortext[cor1]{Corresponding author}

\address[liv]{Institute for Risk and Uncertainty, School of Engineering,
 University of Liverpool\\ Brownlow Hill, Liverpool L69 3GH, United Kingdom}

\begin{abstract}

Gaussian process emulators of computationally expensive computer codes provide
fast statistical approximations to model physical processes. The training of
these surrogates depends on the set of design points chosen to run the
simulator. Due to computational cost, such training set is bound to be limited
and quantifying the resulting uncertainty in the hyper-parameters of the
emulator by uni-modal distributions is likely to induce bias. In order to
quantify this uncertainty, this paper proposes a computationally efficient
sampler based on an extension of Asymptotically Independent Markov Sampling, a
recently developed algorithm for Bayesian inference. Structural uncertainty of
the emulator is obtained as a by-product of the Bayesian treatment of the hyper-
parameters. Additionally, the user can choose to perform stochastic optimisation
to sample from a neighbourhood of the Maximum a Posteriori estimate, even in the
presence of multimodality. Model uncertainty is also acknowledged through
numerical stabilisation measures by including a nugget term in the formulation
of the probability model. The efficiency of the proposed sampler is illustrated
in examples where multi-modal distributions are encountered. For the purpose of
reproducibility, further development, and use in other applications the code
used to generate the examples is freely available for download at
\url{https://github.com/agarbuno/paims_codes}.

\end{abstract}


\end{frontmatter}

\section{Introduction}

Computationally expensive computer codes are frequently needed to implement
mathematical models which are assumed to be reliable approximations to physical
processes. Such simulators often require intensive use of computational
resources that makes them inefficient if further exploitation of the code is
needed, \eg optimisation, uncertainty propagation and sensitivity
analysis \citep{Forrester2008, Kennedy2001}. For this reason, surrogate models
are needed to perform fast approximations to the output of demanding simulators
and enable efficient exploration and exploitation of the input space. In this
context, Gaussian processes are a common choice to build statistical surrogates
-also known as {\em emulators}- which allow to take into account the uncertainty
derived from the inability to evaluate the original model in the whole input
space. Gaussian processes have become popular in recent years due to their
ability to fit complex mappings between outputs and inputs by means of a 
non-parametric hierarchical structure. Such applications are found, amongst many
other areas, in Machine Learning \citep{Rasmussen2006}, Spatial Statistics
\citep{Cressie1993a} (with the name of Kriging), likelihood-free Bayesian
Inference \citep{Wilkinson2014} and Genetics \citep{Kalaitzis2011}.

To build an emulator, a number of runs from the simulator is needed, but due to
computing limitations only a small amount of evaluations can be performed. With
a small amount of data, it is possible that the uncertainty of the parameters of
the model cannot be described by a clearly uni-modal distribution. In such
scenarios, Model Uncertainty Analysis \citep{Draper1995} is capable of setting a
proper framework in which we acknowledge all uncertainties related to the
idealisations made through the modelling assumptions and the available, albeit
limited information. To this end, {\em hierarchical modelling} should be
taken into account. This corresponds to adding a layer of structural uncertainty
to the assumed emulator either in a continuous or discrete manner \citep[see][\S
4]{Draper1995}. In the case of Gaussian processes, continuous structural
uncertainty can be accounted for as a natural by-product from a Bayesian
procedure. Hence, this is pursued in this work by focusing on samplers capable
of exploring multi-modal distributions.

In order for the Gaussian process to be able to replicate the relation between
inputs and outputs and make predictions, a training phase is necessary. Such
training involves the estimation of the parameters of the Gaussian process from
the data collected by running the simulator. These parameters are referred to as
{\em hyper-parameters}. The selection of the hyper-parameters is usually done by
using Maximum Likelihood estimates (MLE), or their Bayesian counterpart Maximum
a Posteriori estimates (MAP) \citep{Oakley1999, Rasmussen2006}, or by sampling
from the posterior distribution \citep{Williams1996} in a fully Bayesian manner.

In this paper we assume a scenario where the task of generating new runs from
the simulator is prohibitive. Such limited information is not enough to
completely identify either a candidate or a region of appropriate candidates for
the hyper-parameters. In this scenario, traditional optimisation routines
\citep{Nocedal2004} are not able to guarantee global optima when looking for the
MLE or MAP, and a Bayesian treatment is the only option to account for all the
uncertainties in the modelling. In the literature, however, it is common to see
that MLE or MAP alternatives are preferred \citep{Kennedy2001, Gibbs1998} due to
the numerical burden of maximising the likelihood function or because it is
assumed that Bayesian integration will not produce results worth the effort.
Though it is a strong argument in favour of estimating isolated candidates, in
high-dimensional applications it is difficult to assess if the number of runs of
the simulator is sufficient to produce robust hyper-parameters. Robustness is
usually measured with a prediction-oriented metric such as root-mean-square
error (RMSE) \citep{Kennedy2001a}, ignoring uncertainty and risk assessment of
choosing a single candidate of the hyper-parameters by an inference process with
limited data. In order to account for such uncertainty in the hyper-parameters
when making predictions, numerical integration should be performed. However,
methods as quadrature approximation become infeasible as the number of
dimensions increases \citep{Kennedy2001}. Therefore, an appropriate approach is
to perform Monte Carlo integration \citep{MacKay1998}. This allows to
approximate any integral by means of a weighted sum, given a sample from the
{\em correct} distribution.

In Gaussian processes, as in many other applications of statistics, the target
distribution of the hyper-parameters cannot be sampled directly and one should
resort to Markov Chain Monte Carlo (MCMC) methods \citep{Robert2004}. MCMC
methods are powerful statistical tools but have a number of drawbacks if not
tuned properly, particularly if one wishes to sample from multi-modal
distributions \citep{Neal2001, Hankin2005}. One of such limitations is the
tuning of the proposal distribution, which allows the generation of a candidate
in the chain. This proposal function has to be tuned with parameters that define
its ability to move through the sample space. If an excessively wide spread is
selected, this will produce samples with space-filling properties but which are
likely to be rejected. On the other hand, having a narrower spread will cause an
inefficient exploration of the sample space by taking short updates of the
states of the chain, known in the literature as {\em Random Walk} behaviour
\citep{Neal1993}. In practice it is desirable to use a proposal distribution
which is capable of balancing both extremes. To find an appropriate tuning in
high-dimensional spaces with sets of highly correlated variables can be an
overwhelming task and often MCMC samplers can become expensive due to the long
time needed to reach stationarity \citep{Ching2007}. \cite{Neal1998a} and
\cite{Williams1996} favour the Hybrid Monte Carlo (HMC) method to generate a
sample from the posterior distribution, preventing the random walk behaviour of
traditional MCMC methods. If tuned correctly, the HMC should be able to explore
most of the input space \citep{Liu2008}. Such tuning process is problem-dependent 
and there is no guarantee that the method will sample from all
existing modes, thus failing to adapt well to multi-modal distributions
\citep{Neal2011a}.

This paper proposes a sampler for the hyper-parameters of a Gaussian process
based on recently developed methods for Bayesian inference problems.
Additionally, it uses the Transitional Markov Chain Monte Carlo (TMCMC) method
of \cite{Ching2007} to set a framework for the parallelisation of Asymptotically
Independent Markov Sampling in both the context of hyper-parameter sampling
(AIMS) \citep{Beck} and in stochastic optimisation (AIMS-OPT) \citep{Zuev2013}
reminiscent of Stochastic Subset Optimisation
\citep{Taflanidis2008,Taflanidis2008a}. Such an extension is built using
concepts of Particle Filtering methods \citep{Andrieu2010, Gramacy2009},
Adaptive Sequential Monte Carlo \citep{DelMoral2006, DelMoral2012} and Delayed
Rejection Samplers \citep{Zuev2011, Mira2001a}. AIMS is chosen since it provides
a framework for Sequential Monte Carlo sampling \citep{Neal1996, Neal2001,
DelMoral2006} which automatically chooses the sequence of transitions. Moreover,
it uses most of the information generated in the previous step in the sequence
as opposed to traditional sequential methods, thus building a robust sampler
when applied to multi-modal distributions. \added{Finally, by using the AIMS-OPT
algorithm a solution is built by means of a nested sequence of subsets, which
converges to the optimal solution set. The algorithm can be terminated
prematurely given a previously chosen accuracy threshold, thus providing a set
of nearly optimal solutions. Whether it is composed by a single element, or a
set of elements whose objective function differs by a negligible quantity, a
full characterisation of the optimal solution is achieved. This contrasts with
the capabilities of other stochastic optimisation schemes such as particle swarm
optimisation or genetic algorithms \citep{Schneider2007}.}

By selecting the hyper-parameters using the AIMS-OPT framework the effect is
twofold. First, the uncertainty inherent to the specification of the 
hyper-parameters is embedded in the set of suboptimal approximations to the
solution. This uncertainty, expressed in a mixture of Gaussian process
emulators, yields a robust surrogate where model uncertainty is accounted for.
Second, computational implementation deficiencies of the inference procedure in
Gaussian processes is overcome by incorporating stabilising approaches exposed
in the literature as in \cite{Ranjan2011, Andrianakis2012} but in a Bayesian
framework. The problem is therefore treated from both a probabilistic and an
optimisation perspective.

The paper is organised as follows. In Section \ref{sec:gps}, a brief introduction
to the Gaussian processes and their treatment by Bayesian inference is
discussed. Section \ref{sec:aims} presents both the AIMS algorithm and the
proper generalisation for a parallel implementation. Section
\ref{sec:implementation} discusses several aspects of the computational
implementation of the algorithm and their effect on the modelling assumptions.
The efficiency and robustness of the proposed sampler are discussed in Section
\ref{sec:numerical} with some illustrative examples. Concluding remarks are
given in Section \ref{sec:conclusions}.

%------------------------------------------------

\section{Gaussian processes} \label{sec:gps}

Let $X = \{ \vec{x}_1, \ldots, \vec{x}_n \} $ be the set of trials run by the
simulator where $\vec{x}_i \in \mathbb{R}^p$ denotes a given configuration for
the model. The set $X$ will be referred to as the set of {\em design points}.
Let $ \vec{y} = \{ y_1, \ldots, y_n \}$ be the set of outputs observed for the
design points. The pair $(\vec{x}_i, y_i)$ will denote the {\em training run}
being used to learn the emulator that approximates the simulator. The emulator
is assumed to be a real-valued mapping $\eta: \mathbb{R}^p \rightarrow
\mathbb{R}$ which is an interpolator of the training runs, \ie $y_i =
\eta(\vec{x}_i)$ for all $i= 1, \ldots, n$. This omits any random error in the
output of the computer code in the observed simulations, that is, the simulator
is deterministic. It is assumed that the output of the simulator can be
represented by a Gaussian process. Therefore, the set of design points is
assumed to have a joint Gaussian distribution where the output satisfies the
structure
\begin{align}
\eta(\vec{x}) = h(\vec{x})^T \boldsymbol\beta + Z(\vec{x}|\sigma^2, \boldsymbol\phi),
\end{align}
where $h(\cdot)$ is a vector of known basis (location) functions of the input,
$\boldsymbol\beta$ is a vector of regression coefficients, and
$Z(\cdot|\sigma^2, \boldsymbol\phi)$ is a Gaussian process with zero mean and
covariance function
\begin{align}
\text{cov}(\vec{x}, \vec{x}'| \sigma^2, \boldsymbol\phi) = \sigma^2 \, k (\vec{x}, \vec{x}'|
\boldsymbol\phi), \label{eq:correlation}
\end{align}
where $\sigma^2$ is the signal noise and $\boldsymbol\phi \in \mathbb{R}^p_+$
denotes the {\em length-scale} parameters of the correlation function $k(\cdot,
\cdot)$.  Note that for a pair of design points$(\vec{x},\vec{x}')$, the
function  $k(\cdot, \cdot | \boldsymbol\phi)$ measures the correlation between
$\eta(\vec{x})$ and $\eta(\vec{x}')$ based on their respective input
configurations. The effect of different values of $\boldsymbol\phi$ in a 
one-dimensional example is depicted in Figure \ref{fig:length_scale}.
\begin{figure}[H]
\centering 
\includegraphics[draft=false,width=.4\linewidth]{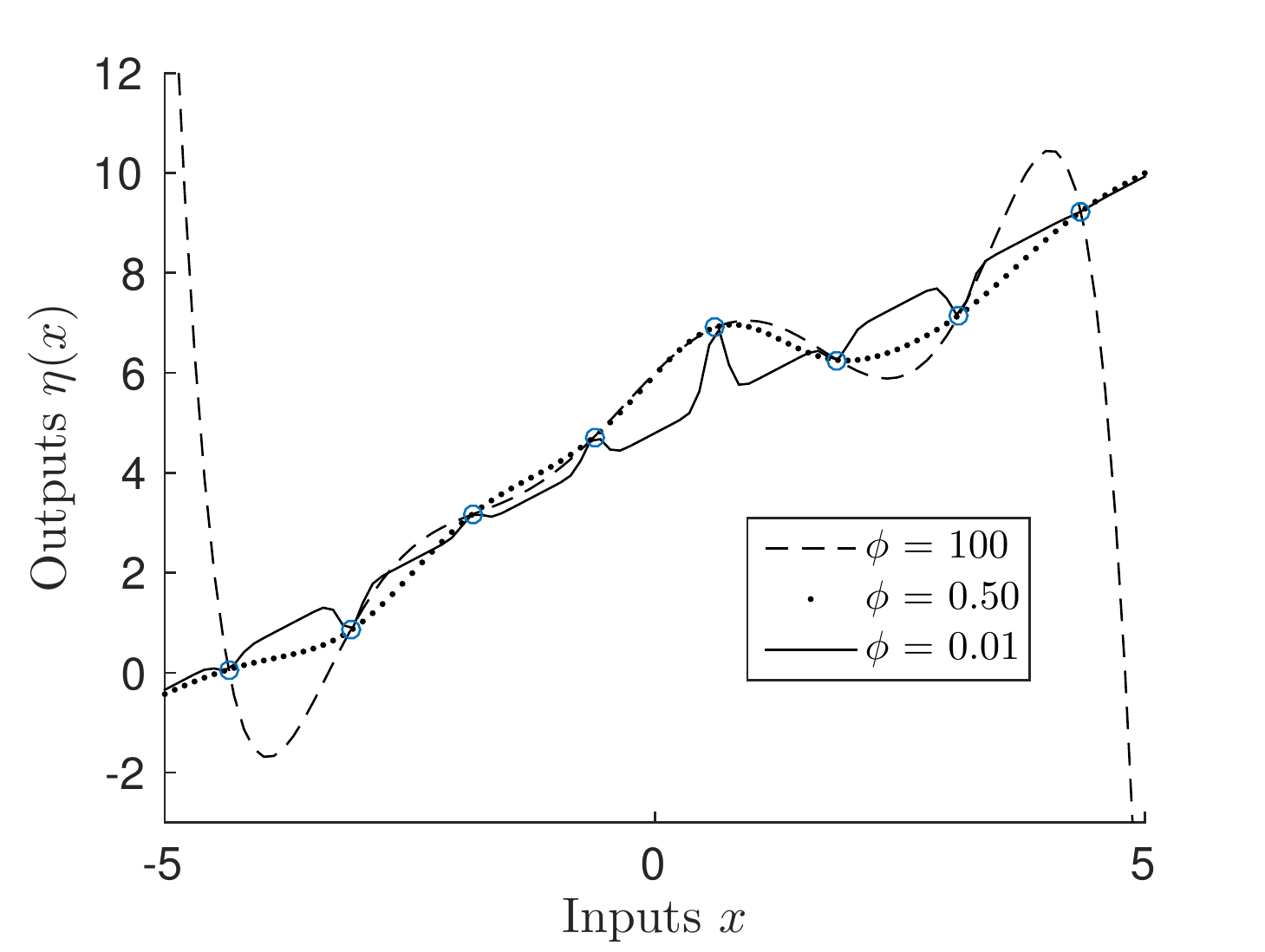}
\caption{The length-scale parameters represent how sensitive is the output of
the simulator to variations in each dimension. The plot corresponds to 8 design
points chosen for the function  $\eta(x) = 5 + x + \cos(x) + .5 \, \sin(3x)$.
For low values of the length-scale parameter the training runs are less
dependent of each other. }
\label{fig:length_scale}
\end{figure}

The role of the correlation function is to measure how close to each other the
design points are, following the assumption that similar input configurations
should produce similar outputs. For its analytical simplicity, interpretation
and smoothness properties, this work uses the squared-exponential correlation
function, namely
\begin{align}
k(\vec{x},\vec{x}'|\boldsymbol\phi) = \exp \left\lbrace - \frac{1}{2} \displaystyle \sum_{i=1}^p
\frac{(x_i - x'_i)^2}{\phi_i}\right\rbrace. 
\end{align}
Note that other authors prefer the parametrisation with $\phi^2_i$ as
denominators. However, this work uses a linear term in the denominator since the
restriction of the length-scale parameters to lie in the positive orthant is
more natural, as weights in the norm used to measure closeness and sensitivity
to changes in such dimensions. \added{Both interpretability and numerical
performance can be improved if the length-scales refer to the same units, which
leads to rescaling all dimensions of the input configurations. In the computer
simulation terminology this translates in utilising experimental designs
restricted to hypercubes, such as Latin hypercube sampling or Sobol sequences.
Design of experiments is an active area of research outside the scope of this
work.}
	
In summary, the output of a design point, given the parameters
$\boldsymbol\beta, \sigma^2$ and $\boldsymbol\phi$, has a Gaussian distribution
\begin{align}
y | \vec{x}, \boldsymbol\beta, \sigma^2, \boldsymbol\phi \sim \mathcal{N}(h(\vec{x})^T
\boldsymbol\beta, \sigma^2 \, k (\vec{x}, \vec{x}'| \boldsymbol\phi) ),
\end{align}
which can be rewritten as the joint distribution of the vector of outputs
$\vec{y}$ conditional on the design points $X$ and hyper-parameters
$\boldsymbol\beta, \sigma^2$ and $\boldsymbol\phi$ as
\begin{align}
{\bf y} | X, \boldsymbol\beta, \sigma^2, \boldsymbol\phi \sim \mathcal{N}(H \boldsymbol\beta, 
\sigma^2 \, K ), \label{eq:normal}
\end{align}
where $H$ is the {\em design matrix} whose rows are the inputs $h(\vec{x}_i)^T$
and $K$ is the correlation matrix with elements $K_{ij} = k(\vec{x}_i,
\vec{x}_j| \boldsymbol\phi)$ for all $i,j = 1, \ldots, n$. 

\subsection{Estimating the hyper-parameters}

The parameters of the process are not known beforehand and this induces
uncertainty in the emulator itself. They can be estimated by Maximum Likelihood
principles, but doing so lacks rigorous uncertainty quantification by
concentrating all the density of the unknown quantities in a single value. The
alternative is to treat them in a fully Bayesian manner and marginalise them
when performing predictions. This way their respective uncertainty is taken into
account. In this scenario, the prediction $y^*$ for a non-observed configuration
$\vec{x}^*$ can be performed with the data available, $\mathcal{D} = ({\bf y},
X)$, and the evidence they shed on the parameters of the Gaussian process.
Therefore, the predictions should be made with the marginalised posterior
distribution
\begin{align}
p(y^*|\vec{x}^*, \mathcal{D}) = \int_\Theta p(y^*|\vec{x}^*, \mathcal{D}, \boldsymbol\theta) \,
p(\boldsymbol\theta | \mathcal{D}) \, d\boldsymbol\theta \label{eq:pred_post},
\end{align}
where $\boldsymbol\theta = (\boldsymbol\beta, \sigma^2, \boldsymbol\phi)$
denotes the complete vector of hyper-parameters. One should note that given the
properties of a collection of Gaussian random variables, a prediction for $y^*$
conditioned in the data and $\boldsymbol\theta$ is also a Gaussian random
variable \citep[see][]{Oakley1999}. As in hierarchical modelling, each possible
value of $\boldsymbol\theta$ defines a specific realisation of a Gaussian
distribution, so it is appropriate to refer to $\boldsymbol\theta$ as the 
hyper-parameters of the Gaussian process.

Due to its computational complexity, the integral in \eqref{eq:pred_post} is
often omitted when making predictions. It is commonly assumed that the MLE of
the likelihood
\begin{align}
\mathcal{L}(\boldsymbol\theta) = p(\vec{y}|X,\boldsymbol\beta, \sigma^2, \boldsymbol\phi), 
\label{eq:likelihood}
\end{align}
or the MAP estimate from the posterior distribution 
\begin{align}
p(\boldsymbol\theta | \mathcal{D}) \propto p({\bf y}|X,\boldsymbol\beta, \sigma^2, \boldsymbol\phi)
\, \, p(\boldsymbol\beta, \sigma^2, \boldsymbol\phi), \label{eq:posterior}
\end{align}
are robust enough to account for all the uncertainty in the modelling. However,
when either the likelihood \eqref{eq:likelihood} is a non-convex function or the
posterior \eqref{eq:posterior} is a multi-modal distribution, conventional
optimisation routines might only find local optima, thus failing to find the
most probable candidate of such distribution. Moreover, by selecting only one
candidate, robustness and uncertainty quantification are lost in the process.
Additionally, there are degenerate cases when it is crucial to estimate the
integral in \eqref{eq:pred_post} by means of Monte Carlo simulation instead of
by proposing a single candidate. As it has been noted by \cite{Andrianakis2012},
two extreme cases for the Gaussian process length-scale hyper-parameters can be
identified. One possibility is for $\phi$ to approach infinity, which makes
every design point dependent of each other; the other, when $\phi$ approaches
the origin where a multivariate regression model becomes the  limiting case. In
the first case, high correlation among all the training runs results in a model
which is not able to distinguish local dependencies. As for the second, it
violates the assumptions that constitute a Gaussian process, by completely
ignoring the correlation structure in the design points to predict the output.
Consequently, if MCMC is performed one can approximate the integrated predictive
distribution in \eqref{eq:pred_post} by means of
\begin{align}
p(y^*|\vec{x}^*, \mathcal{D}) \approx \sum_{i=1}^N w_i \,\, p(y^*|\vec{x}^*, \mathcal{D},
\boldsymbol\theta_i) \label{eq:mix_posterior} , 
\end{align} 
where $\boldsymbol\theta_i$ is obtained through an appropriate sampler, \ie one
capable of sampling from multi-modal distributions. The coefficients $w_i$
denote the weights of each sample generated. Since each term $p(y^*|\vec{x}^*,
\mathcal{D}, \boldsymbol\theta_i)$ corresponds to a Gaussian density function,
the predictions are made by a mixture of Gaussians.

\begin{proposition}

If the emulator output $y^*$ conditional on its configuration vector $\vec{x}^*$
has a posterior density as in \eqref{eq:mix_posterior}, then its mean function
and covariance function can be computed as
\begin{align}
\mu(\vec{x}^*) &= \sum_{i=1}^N w_i \,\, \mu_i(\vec{x}^*), \label{eq:mix_mean}\\
\text{cov}(\vec{x}^*,\vec{x}') &= \sum_{i=1}^N w_i \,\, \left[ (\mu_i(\vec{x}^*) - \mu(\vec{x}^*))
(\mu_i(\vec{x}') - \mu(\vec{x}')) + \text{cov}(\vec{x}^*,\vec{x}' | \boldsymbol\theta_i) \right],
\label{eq:mix_cov}
\end{align}
where $\mu_i(\vec{x}^*)$ denotes the expected value of the likelihood
distribution of $y^*$ conditional on the hyper-parameters $\boldsymbol
\theta_i$, the training runs $\mathcal{D}$ and the input configuration
$\vec{x}^*$.
\end{proposition} 
\begin{proof}
Equality in \eqref{eq:mix_mean} is a direct application of the tower property of
conditional expectation and \eqref{eq:mix_cov} follows from the covariance
decomposition formula using the vector of weights $w_i$ as an auxiliary
probability distribution on the conditioning.
\end{proof}
\noindent From equation \eqref{eq:mix_cov} we can compute the variance, also
known as the prediction error, of an untested configuration $\vec{x}^*$ as
\begin{align}
\replaced{\text{var}}{\sigma^2}(\vec{x^*}) &= \sum_{i=1}^N w_i \,\, ( (\mu_i(\vec{x^*}) - \mu(\vec{x^*}))^2 +
\sigma_i^2(\vec{x^*})). \label{eq:mix_sigma}
\end{align} 
\noindent By doing this, a more robust estimation of the prediction error is
made since it balances the predicted error in one sample with how far the
prediction of such sample is from the overall estimation of the mixture.

\subsection{Prior distributions}

In order to perform a Bayesian treatment for the prediction task in equation
\eqref{eq:pred_post} the prior distribution $ p(\boldsymbol\beta, \sigma^2,
\boldsymbol\phi)$ in equation \eqref{eq:posterior} has to be specified. Weak
prior distributions are commonly used for $\boldsymbol\beta$ and $\sigma^2$
\citep{Oakley1999}. Such weak prior has the form
\begin{align}
 p(\boldsymbol\beta, \sigma^2, \boldsymbol\phi) \propto	\frac{p(\boldsymbol\phi)}{\sigma^2},
\label{eq:prior}
\end{align}
where it is assumed a priori that both the covariance and the mean 
hyper-parameters are independent. Even more, $\boldsymbol\beta$ and $\sigma^2$
are assumed to have an improper non-informative distribution.

As for the length-scale hyper-parameter $\boldsymbol\phi$, a prior distribution
$p(\boldsymbol\phi)$ is still needed. In this case the reference prior
\citep[studied by][]{Berger1992a, Berger2009} sets an objective framework to
account for the uncertainty of $\boldsymbol\phi$, thus avoiding any potential
bias induced by the modelling assumptions. This prior is built based on
Shannon's expected information criteria and allows the use of a prior
distribution in a setting where no previous knowledge is assumed. That way, the
training runs are the only source of information for the inference process.
Additionally, the reference prior is capable of ruling out subspaces of the
sample space of the hyper-parameters \citep{Andrianakis2011}, thus reducing
regions of possible candidates of Gaussian distributions in the mixture model in
equation \eqref{eq:mix_posterior}. Since this provides an off-the-shelf
framework for the estimation of the hyper-parameters, the reference  prior
developed by \cite{Paulo2005} is used in this work. However, there are no known
analytical expressions for its derivatives which limits its application to MCMC
samplers that use gradient information. Note that there are other possibilities
available for the prior distribution of $\boldsymbol\phi$. Examples of these are
the log-normal or log-Laplacian distributions, which can be interpreted as a
regularisation in the norm of the parameters. \cite{Andrianakis2011} suggest a
decaying prior. Another option is to elicit prior distributions from expert
knowledge as in \cite{Oakley2002a}. 

\subsection{Marginalising the nuisance hyper-parameters}

The nature of the hyper-parameters $\boldsymbol\beta, \sigma^2$ and
$\boldsymbol\phi$ is potentially different in terms of scales and dynamics, as
seen and explained in Figure \ref{fig:dynamics}. It is possible to cope with
this limitations by using a Gibbs sampling framework, but it is well-known that
such sampling scheme can be inefficient if it is used for multi-modal
distributions in higher dimensions. Analogously, a Metropolis-Hastings sampler
can also be overwhelmed.
\begin{figure}[H]
\centering
\subfloat[Correlations]{\includegraphics[draft=false,width=.32\linewidth]{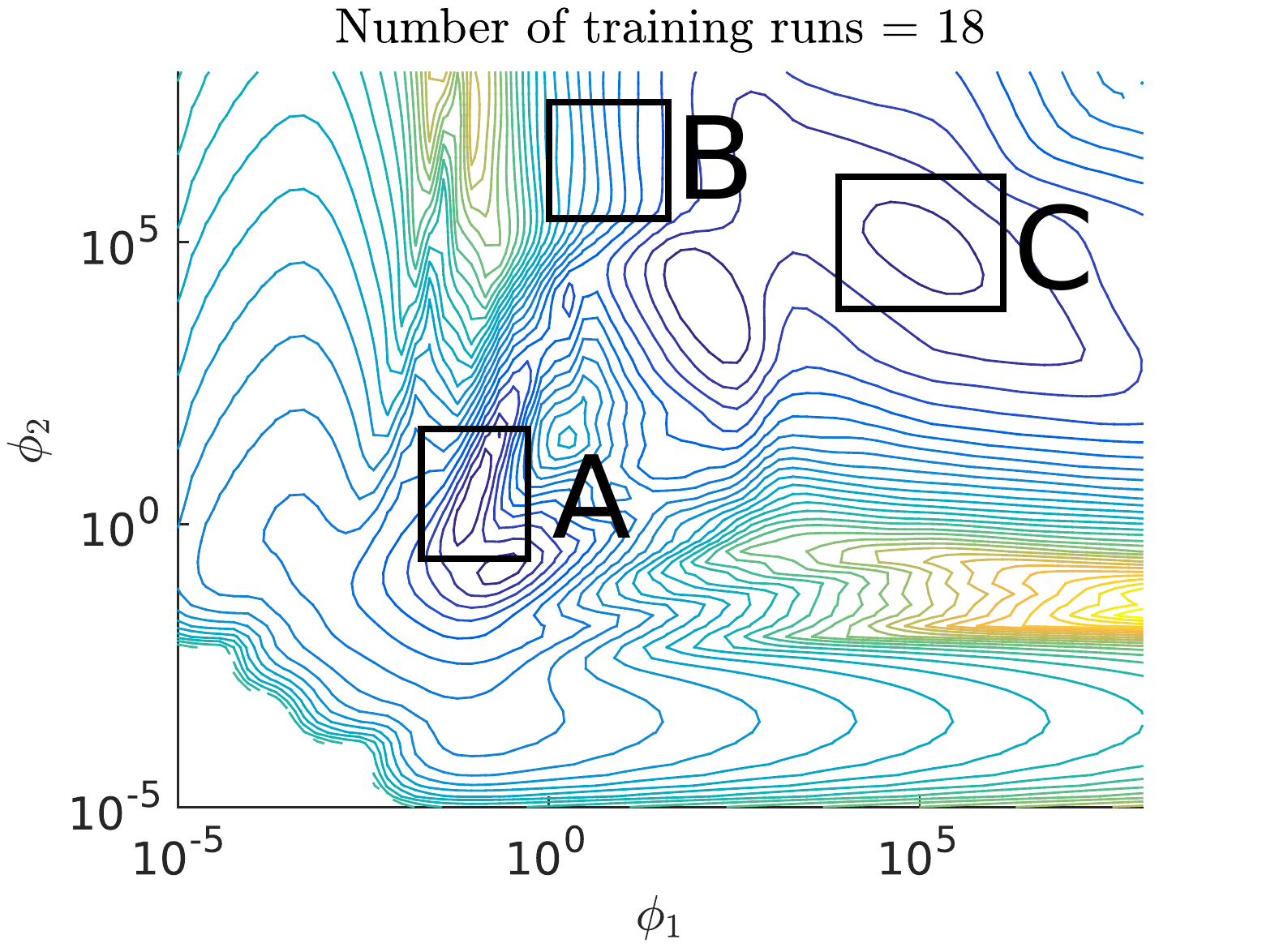} \label{subfig:dynamics01}} \qquad
\subfloat[Scales]{\includegraphics[draft=false,width=.32\linewidth]{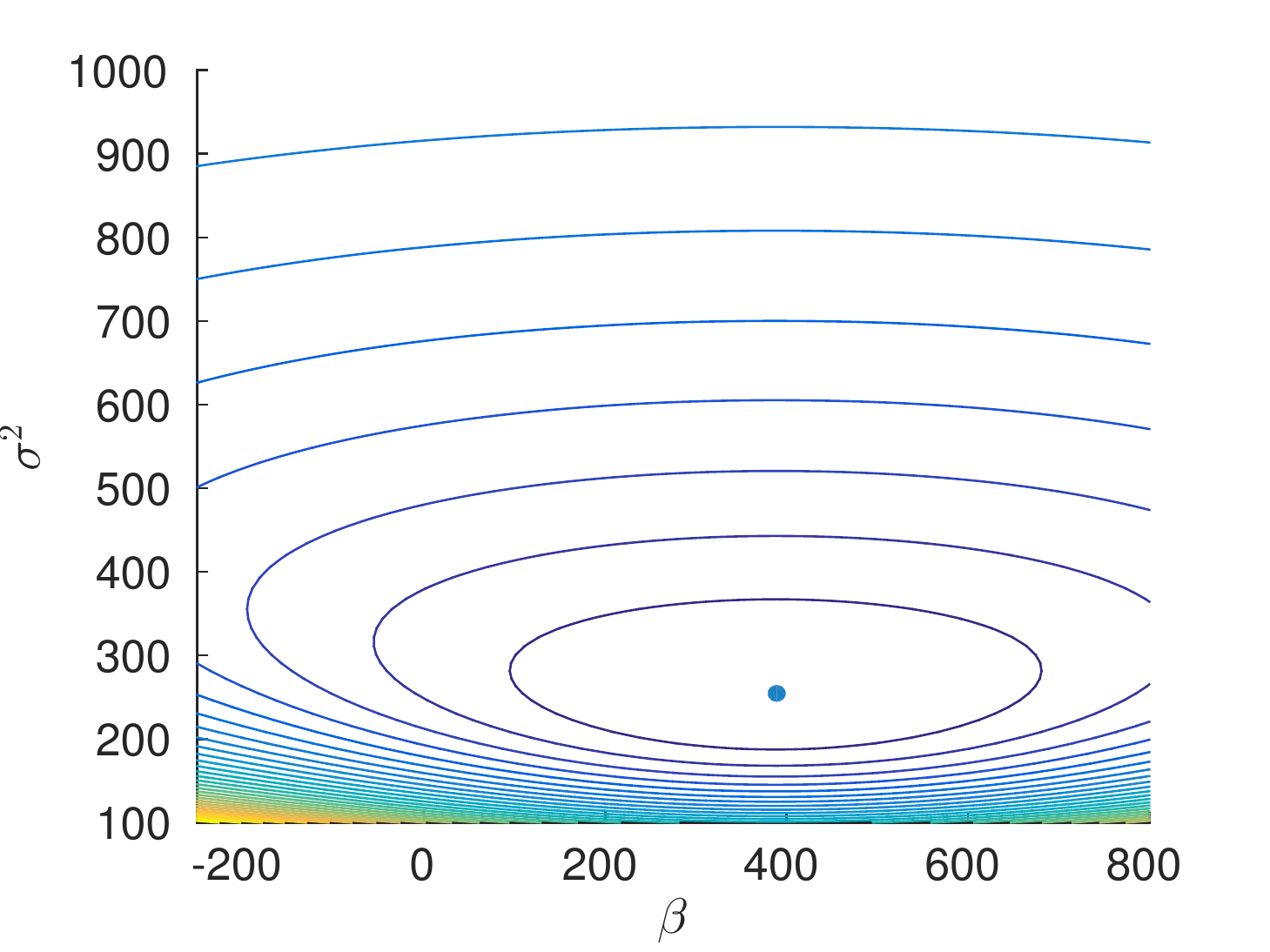}
\label{subfig:dynamics02}}
\caption{In \ref{subfig:dynamics01}, different dynamics of the hyper-parameters
for the log-posterior distribution of test function \ref{subsec:branin} are
shown: {\bf A.} corresponds to positive correlation. {\bf B.} corresponds to an
independent region. {\bf C.} corresponds to negative correlation. In
\ref{subfig:dynamics02}, the marginal log-posterior function of the same example
with $h(x)=1$, presents the same contour level for a wide range of
$\boldsymbol\beta$. Thus, the hyper-parameters exhibit very different scales.
The dot represents the minimum of the corresponding function.}
\label{fig:dynamics}
\end{figure}
\noindent Another alternative is to focus on $\boldsymbol\phi$ and perform the
inference in the correlation function. This is done by regarding
$\boldsymbol\beta$ and $\sigma^2$ as nuisance parameters and integrating them
out from the posterior distribution \eqref{eq:posterior}. The modelling
assumptions in the training runs and the prior distribution, equations
\eqref{eq:normal} and \eqref{eq:prior} respectively, allow to identify a
Gaussian-inverse-gamma distribution for $\boldsymbol\beta$ and $\sigma^2$, which
can be shown to yield the integrated posterior distribution
\begin{align}
p(\boldsymbol\phi | \mathcal{D}) \propto p(\boldsymbol\phi) \, (\hat{\sigma}^2)^{-\frac{n-p}{2}} \,
|K|^{-\frac{1}{2}} \, |H^T K^{-1} H|^{-\frac{1}{2}}, \label{eq:integrated_post}
\end{align}
where 
\begin{align}
\hat{\sigma}^2 &= \frac{\vec{y}^T \, \left( K^{-1} - K^{-1}H(H^T K^{-1} H)^{-1}H^TK^{-1} \right) \,
\vec{y}}{n-p-2},
\end{align}
and
\begin{align}
\hat{\boldsymbol\beta} &= (H^T K^{-1} H)^{-1}H^TK^{-1} \vec{y},
\end{align}
are estimators of the signal noise $\sigma^2$ and regression coefficients
$\boldsymbol\beta$ \cite[see][for further details]{Oakley1999}. Additionally,
the predictive distribution conditioned on the hyper-parameters follows a
Gaussian distribution with mean and correlation functions
\begin{align}
\mu(\vec{x}^*| \boldsymbol\phi) =& \,\, h(\vec{x}^*)^T \hat{\boldsymbol\beta} + t(\vec{x}^*)^T
K^{-1} (\vec{y} - H\hat{\boldsymbol\beta}), \\
\text{corr}(\vec{x}^*,\vec{w}^*|\boldsymbol\phi) =& \,\, k(\vec{x}^*, \vec{w}^* | \boldsymbol\phi) -
t(\vec{x}^*)^T K^{-1} \, t(\vec{w}^*) \, + \nonumber \\ 
& \left(h(\vec{x}^*)^T - t(\vec{x}^*)^T K^{-1} H \right) (H^T K^{-1} H)^{-1} \left(h(\vec{w}^*)^T -
t(\vec{w}^*)^T K^{-1} H \right)^T,
\end{align}
where $\vec{x}^*$, $\vec{w}^*$ denote a pair of test configurations and
$t(\vec{x}^*)$ denotes the vector obtained by computing the covariance of the
new proposal with every design point $t(\vec{x}) =
(k(\vec{x},\vec{x}_1|\boldsymbol\phi), \ldots k(\vec{x}, \vec{x}_n|
\boldsymbol\phi))^T$. Note that both estimators depend only on the correlation
function hyper-parameters $\boldsymbol\phi$ since both $\boldsymbol\beta$ and
$\sigma^2$ have been integrated out. Considerations of when it is appropriate to
integrate out the hyper-parameters in a model has been discussed by
\cite{MacKay1996}. In the Gaussian process context it gains additional
significance since it allows the development of appropriate MCMC samplers
capable of overcoming the dynamics of different sets of hyper-parameters.

In the light of the above discussion, this work focuses on the inference drawn
from the correlation function $k(\cdot,\cdot)$ in equation
\eqref{eq:correlation}, since the structure of dependencies of the training runs
to predict the outputs is recovered by it. The main assumption is that the mean
function hyper-parameter $\boldsymbol\beta$ contains minor information on the
structural dependencies of the data, relative to the correlation function 
hyper-parameters, which would prevent the use of integrated likelihoods
\citep[see][for further discussion]{Berger1999}. If prior information is
available, then an additional effort can be made on eliciting an appropriate
mean function for the Gaussian process emulator.  Such information can be
related to expert knowledge of the simulator which eventually allows  the
analyst to build a better mean function by adding significant regression
covariates \cite[see][for a detailed discussion]{Vernon2010}.

%------------------------------------------------

\section{AIMS Framework} \label{sec:aims}

Hyper-parameter marginalisation by means of Monte Carlo methods in Gaussian
processes is usually performed by Hybrid Monte Carlo methods \citep{Neal1998a,
Williams1996} which are capable of suppressing the Random Walk behaviour of MCMC
samplers if tuned correctly. In this work, the sampling of the hyper-parameters
is done by means of Asymptotically Independent Markov Sampling (AIMS)
\citep{Beck}. This method combines techniques developed for Bayesian inference
such as Importance Sampling and Simulated Annealing \citep{Kirkpatrick1983} to
sample from the posterior distribution as done by other MCMC algorithms.
Additionally, AIMS can also be adapted for global optimisation (AIMS-OPT)
\citep{Zuev2013} in a fashion of the traditional simulated annealing method for
stochastic optimisation. Let the problem be
\begin{align}
\min_{\boldsymbol\phi \in \Phi} \, \mathcal{H}(\boldsymbol\phi | \mathcal{D}), 
\label{eq:objective}
\end{align}
where $\mathcal{H}(\phi | \mathcal{D})$ denotes the negative log-posterior
distribution conditional on the set of training runs $\mathcal{D}$. Let the set
of optimal solutions to the optimisation problem above be
\begin{align}
\Phi^* = \left\lbrace \boldsymbol\phi \in \Phi \,:\, \boldsymbol\phi = \arg \min_{\boldsymbol\phi
\in \Phi} \, \mathcal{H}(\boldsymbol\phi | \mathcal{D} ) \right\rbrace , \label{eq:opt_set}
\end{align}
where $|\Phi^*| \geq 1$. This formulation acknowledges the presence of multiple
global optima in the posterior distribution conditional on the training runs. It
is important to note that using the logarithm of the posterior distribution
reduces the overflow in the computation of the equation
\eqref{eq:integrated_post}, which is likely to arise due to ill-conditioning of
the matrix $K$ \citep{Neal2003a}.

In this context, AIMS-OPT is capable of producing samples by means of a sequence
of nested subsets $\Phi_{k+1} \subseteq \Phi_{k}$ that converges to the set of
optimal solutions $\Phi^*$. Thus, if the algorithm is terminated in a premature
step, a set of sub-optimal approximations to \eqref{eq:opt_set} will be
recovered. Let $\{p_k(\boldsymbol\phi | \mathcal{D})\}_{k = 1}^{\infty}$ be the
sequence of density distributions such that
\begin{align}
p_k(\boldsymbol\phi | \mathcal{D}) &\propto p(\boldsymbol\phi | \mathcal{D})^{1/\tau_k} = \exp \left\lbrace - \mathcal{H}(\boldsymbol\phi | \mathcal{D} )/\tau_k \right\rbrace,
\end{align} 
for a sequence of monotonically decreasing temperatures $\tau_k$. By tempering
the distributions in this manner, the samples obtained in the first step of the
algorithm are approximately distributed as a uniform random variable over a {\em
practical support}; while in the last annealing level, they are distributed
uniformly on the set of optimal solutions, namely
\begin{align}
\lim_{\tau \rightarrow \infty} p_\tau(\boldsymbol\phi| \mathcal{D}) &= U_\Phi(\boldsymbol\phi),
\label{eq:meta_prior}\\
\lim_{\tau \rightarrow 0} p_\tau(\boldsymbol\phi| \mathcal{D}) &= U_{\Phi^*}(\boldsymbol\phi),
\end{align}
where $U_A(\boldsymbol\phi)$ denotes a uniform distribution over the set $A$ for
every $\boldsymbol\phi \in A$.

\subsection{Annealing at level \textit{k}}

The general framework for the AIMS-OPT algorithm is presented, focusing on how
to sample from the hyper-parameter space at level $k$ based on the sample of the
previous level. Let $\boldsymbol\phi_1^{(k-1)}, \ldots,
\boldsymbol\phi_N^{(k-1)}$ be samples of the hyper-parameters distributed as
$p_{k-1}(\boldsymbol\phi)$ at level $k-1$. For notational simplicity, the
conditional on $\mathcal{D}$ will be omitted from $p_{k-1}(\cdot)$, however the
training runs are crucial to build statistical surrogates. The objective is to
use a kernel such that $p_k(\cdot)$ is the stationary distribution of the Markov
chain. Let $\mathcal{P}_k$ denote such Markov transition kernel, which satisfies
the continuous Chapman-Kolmogorov equation
\begin{align}
p_k(\boldsymbol\phi) \, d\boldsymbol\phi = \int_\Phi \, \mathcal{P}_k(d\boldsymbol\phi|
\boldsymbol\xi) \, p_k(\boldsymbol\xi) \, d\boldsymbol\xi, \label{eq:stationarity}
\end{align}
where $ p_k(d\boldsymbol\phi) = p_k(\boldsymbol\phi) \, d\boldsymbol\phi$
denotes the probability measure. By applying importance sampling using the
distribution at the previous annealing level, equation \eqref{eq:stationarity}
can be approximated as
\begin{align}
p_k(\boldsymbol\phi) \, d\boldsymbol\phi &= \int_\Phi \, \mathcal{P}_k(d\boldsymbol\phi|
\boldsymbol\xi) \, \frac{p_k(\boldsymbol\xi)}{p_{k-1}(\boldsymbol\xi)} \, p_{k-1}(\boldsymbol\xi) \,
d\boldsymbol\xi \nonumber \\
&\approx \sum_{j=1}^N \, \mathcal{P}_k(d\boldsymbol\phi| \boldsymbol\phi^{(k-1)}_j) \,
\overline{\omega}_j^{(k-1)} = \hat{p}_{k,N}(d\boldsymbol\phi), \label{eq:approx}
\end{align}
where $\hat{p}_{k,N}(\cdot)$ is used as the {\em global} proposal distribution
for a candidate in the chain and
\begin{align}
\omega^{(k-1)}_j &= \frac{p_{k}\left( \boldsymbol\phi^{(k-1)}_j \right)}{p_{k-1}\left(
\boldsymbol\phi^{(k-1)}_j \right)} \propto \exp \left\lbrace - \mathcal{H}\left( 
\boldsymbol\phi^{(k-1)}_j | \mathcal{D} \right) \left( \frac{1}{\tau_k} -
\frac{1}{\tau_{k-1}}\right)\right\rbrace, \\
\overline{\omega}^{(k-1)}_j &= \frac{\omega^{(k-1)}_j}{\sum_{j=1}^N \omega^{(k-1)}_j},
\end{align}
are the importance weights and the normalised importance weights respectively.
Note that for computing $\overline{\omega}^{(k-1)}_j$ the normalising constant
of the integrated posterior  distribution \eqref{eq:integrated_post} is not
needed.
	
The proposals of candidates for the chain are done in two steps. In the first
step, a candidate is drawn as an update from a random {\em marker} from the
sample of the previous annealing level, checking whether it is accepted or not.
If the local candidate is rejected by a Random Walk Metropolis-Hastings
evaluation, then the chain remains invariant, $\boldsymbol\phi^{(k)}_{i+1} =
\boldsymbol\phi^{(k)}_{i}$, and another marker is selected at random. In the
second step, given the candidate has been accepted as a local proposal, such
candidate is considered as being drawn from the approximation in
\eqref{eq:approx} and accepted in an Independent Metropolis-Hastings framework,
hence called a global candidate for the chain. Let $q_k(\cdot|\cdot)$ denote the
symmetric transition distribution used for local proposals for the Markov chain.
The subscript $k$ accounts for the adaptive nature of the transition steps in
each annealing level. Thus, the kernel distribution of the Random Walk, which
leaves the  intermediate density invariant, can be written as
\begin{align}
\mathcal{P}_k(d\boldsymbol\phi| \boldsymbol\xi) =& \, q_k(\boldsymbol\phi| \boldsymbol\xi) \, \min
\left\lbrace 1, \frac{p_k(\boldsymbol\phi)}{p_k(\boldsymbol\xi)}\right\rbrace d\boldsymbol\phi \, +
\, \left( 1- \alpha_k(\boldsymbol\xi) \right) \, \delta_{\boldsymbol\xi}(d\boldsymbol\phi),
\end{align}
where $\delta_{\boldsymbol\xi}(d\boldsymbol\phi)$ denotes a delta density and
$\alpha_k(\boldsymbol\xi)$ is the probability of accepting the transition from
$\boldsymbol\xi$ to $\Phi\backslash \{\boldsymbol\xi\}$. It follows from
\eqref{eq:approx} that the approximated stationary condition of the target
distribution at annealing level $k$ can be written as
\begin{align}
\hat{p}_{k,N}(\boldsymbol\phi) &= \sum_{j=1}^N \overline{\omega}_j^{(k-1)} \,
q_k\left(\boldsymbol\phi \left\lvert \boldsymbol\phi^{(k-1)}_j \right. \right) \, \alpha_k^{l}\left(
\boldsymbol\phi \left\vert \boldsymbol\phi^{(k-1)}_j \right. \right), 
\label{eq:approx_comp}
\end{align}
with 
\begin{align}
\alpha_k^{l}\left( \boldsymbol\xi \left\vert \boldsymbol\phi \right. \right) = \min \left\lbrace 1,
\frac{p_k(\boldsymbol\xi)}{ p_k\left(\boldsymbol\phi \right)}\right\rbrace,
\end{align}
the probability of accepting the local transition; whereas 
\begin{align}
\alpha_k^{g}\left( \boldsymbol\xi \left\vert \boldsymbol\phi \right. \right) = \min \left\lbrace 1,
\frac{p_k(\boldsymbol\xi) \, \hat{p}_{k,N}(\boldsymbol\phi)}{ p_k\left(\boldsymbol\phi \right) \,
\hat{p}_{k,N}(\boldsymbol\xi) }\right\rbrace , \label{eq:global_accep}
\end{align}
denotes the probability of accepting such candidate for the Markov chain, hence
accepting a global transition \citep[see][for a detailed discussion]{Zuev2013}.
This leads to the following two algorithms for each level in the annealing
sequence.

\begin{algorithm}[H]
\SetKwInOut{Input}{Input}\SetKwInOut{Output}{Output}

\Input{}
\begin{minipage}{.92\linewidth}
\vspace{-1.5mm}
\begin{itemize}[label=$\diamond$]\setlength{\itemsep}{-1mm}
\item $\boldsymbol\phi^{(k-1)}_1, \ldots, \boldsymbol\phi^{(k-1)}_N \sim p_{k-1}(\boldsymbol\phi)$,
generated at previous level;
\item $\boldsymbol\phi_1^{(k)} \in \Phi\backslash \left\lbrace \boldsymbol\phi^{(k-1)}_1, \ldots,
\boldsymbol\phi^{(k-1)}_N\right\rbrace$, initial state of the chain;
\item $q_k(\boldsymbol\phi|\boldsymbol\xi)$, symmetric local proposal;
\end{itemize}
\vspace{-1.5mm}
\end{minipage}

\Output{\\
\vspace{-2.5mm}
\begin{itemize}[label=$\diamond$]\setlength{\itemsep}{-2mm}
\item $\boldsymbol\phi^{(k)}_1, \ldots, \boldsymbol\phi^{(k)}_N \sim p_{k}(\boldsymbol\phi)$;
\end{itemize}
\vspace{-1.5mm}
}

\BlankLine

\For{$i\leftarrow 2$ \KwTo $n-1$}{
\begin{minipage}{.85\linewidth}
\begin{enumerate}[label={(\arabic*)}, leftmargin=*]
\item Generate a local candidate using the previous level samples as ``markers''
\begin{align}
\boldsymbol\xi \sim & \, Q_k\left( \boldsymbol\xi \left\rvert \boldsymbol\phi^{(k-1)}_1, \ldots,
\boldsymbol\phi^{(k-1)}_n \right. \right) \nonumber \\
 & = \sum_{j=1}^N \, \overline{\omega}_j^{(k-1)} \, q_k\left(\boldsymbol\xi\left\lvert
\boldsymbol\phi^{(k-1)}_j\right.\right)
\end{align}
\vspace{-7mm}
\begin{enumerate}
\item Select index $j$ with probability proportional to importance weights $\omega_1^{(k-1)},
\ldots, \omega_N^{(k-1)}$.
\item Generate candidate from the local proposal distribution 
\vspace{-2.5mm}
\begin{align}
\boldsymbol\xi \sim q_k\left( \boldsymbol\xi \left\lvert \boldsymbol\phi^{(k-1)}_j \right. \right)
\;
\end{align}
\vspace{-7.5mm}
\item Accept $\boldsymbol\xi$ as a local candidate with probability 
\vspace{-2.5mm}
\begin{align}
\alpha^{l}_k\left( \boldsymbol\xi \left\vert \boldsymbol\phi^{(k-1)}_j \right. \right) \;
\end{align}
\end{enumerate}

\item Update $\boldsymbol\phi_i^{(k)} \rightarrow \boldsymbol\phi_{i+1}^{(k)}$ by accepting or
rejecting $\boldsymbol\xi$ using Algorithm \ref{alg:global_accep}.
\end{enumerate}

\end{minipage}
}
\caption{AIMS-OPT at annealing level $k$}\label{alg:annealing_level}
\end{algorithm}
 
\newpage  
 
\begin{algorithm}[H]
\eIf{$\boldsymbol\xi$ was accepted as local candidate}{
Accept $\boldsymbol\xi$ as a global transition with probability
\vspace{-2.5mm}
\begin{align}
\alpha^{g}_k\left( \boldsymbol\xi \left\vert \boldsymbol\phi^{(k)}_i \right. \right) \;
\end{align}
\vspace{-5.5mm}
}{
Leave the chain invariant
\vspace{-2.5mm}
\begin{align}
\boldsymbol\phi_{i+1}^{(k)} = \boldsymbol\phi_{i}^{(k)} \;
\end{align}
\vspace{-5.5mm}
}
\caption{Global acceptance of $\boldsymbol\xi$}\label{alg:global_accep}
\end{algorithm}

\medskip
According to Algorithm \ref{alg:annealing_level} the initialising step should
also be provided for the annealing level. In practical implementations it is
suggested that it should be considered to be $\boldsymbol\phi^{(k)}_1 \sim
q_k(\boldsymbol\phi | \boldsymbol\phi_j^{(k-1)}) $ where \linebreak 
\noindent $j = \arg \max_i \, \overline{\omega}_i^{(k-1)}$, \ie the sample with the
largest normalised importance weight.

\subsection{Adaptive proposal distribution and temperature scheduling}

Even though a Random Walk is performed in every local proposal, AIMS-OPT
performs efficient sweeping of the sample space by producing candidates from
neighbourhoods of the markers from the previous annealing level
$\{\boldsymbol\phi_j^{(k-1)}\}_{j=1}^N$. This is accomplished if the transition
distribution $q_k (\boldsymbol\phi | \boldsymbol\phi^{(k-1)}_j)$ uses an
appropriate proposal distribution where sampling is to be realised; namely, the
level curves of the tempered distribution. To be able to cope with the 
non-negative restriction and to neglect the effect of the scales on each
dimension, the transitions are performed in the log-space of the length-scale
parameters $\boldsymbol\phi$, as suggested by \cite{Neal1997}. The symmetric
transition distribution proposed is a Gaussian distribution for such 
log-parameters. That is, each local candidate will be distributed as

\begin{align}
\boldsymbol\xi \sim \mathcal{N}\left( \boldsymbol\xi \left\lvert \boldsymbol\phi_j^{(k-1)} , c_k
\Sigma_k \right. \right), \label{eq:gauss_local}
\end{align}
where $c_k$ is a decaying parameter for the spread of the proposal, \ie
$c_k = \nu \, c_{k-1}$ with $\nu \in (0, 1)$ commonly chosen as $\nu=1/2$
\citep{Zuev2013}. The matrix $\Sigma_k$ denotes the covariance matrix for 
log-parameters where typical choices can be the identity matrix $I_{p\times p}$, a
diagonal matrix or a symmetric positive definite matrix. We propose the use of
the weighted covariance matrix estimated from the sample and their importance
weights of the previous level $( \overline{\omega}_1^{(k-1)},
\boldsymbol\phi_1^{(k-1)}) , \ldots, ( \overline{\omega}_N^{(k-1)},
\boldsymbol\phi_N^{(k-1)})$. By doing so, the scale and directions of the
ellipsoids of the Gaussian steps are learned as in Adaptive Sequential Monte
Carlo methods \citep{Haario2001, Fearnhead2013} from the information gathered
from the previous level in the sequence.

The annealing sequence and its effective exploration of the sample space is
dictated by the temperature $\tau_k$ of the intermediate distributions.
Moreover, it defines how different is one target distribution from the next one,
so the effectiveness of the sample as markers from the previous annealing level
depends strongly on how the scheduling is performed. It is clear that abrupt
changes lead to rapid deterioration of the sample, whilst low paced changes
could produce unnecessary steps in the annealing schedule. In order to cope with
this compromise, \cite{Zuev2013} used the {\em effective sampling size} to
determine the value of the next temperature in the process. That is solving for
$\tau_{k}$, when a sample from level $k-1$ has been produced, in
\begin{align}
\frac{\sum_{j=1}^n \exp \left\lbrace -2 \mathcal{H} (\boldsymbol\phi_j^{(k-1)} ) \, \left(
\frac{1}{\tau_k} - \frac{1}{\tau_{k-1}}\right)  \right\rbrace }{\left( \sum_{j=1}^n \exp
\left\lbrace - \mathcal{H} (\boldsymbol\phi_j^{(k-1)} ) \, \left( \frac{1}{\tau_k} -
\frac{1}{\tau_{k-1}}\right)  \right\rbrace\right)^2} = \frac{1}{\gamma n},
\end{align}
where $\gamma$ defines a threshold for the proportion of the sample to be as
effective from the importance sampling. Note that the value of $\gamma$ defines
additionally how many annealing steps will be performed. As suggested from
\cite{Zuev2013} a value of 1/2 is used for such parameter.

\subsection{Stopping condition}

If the temperature continues to drop along the sequence of intermediate
distributions, eventually an {\em absolute zero} $\tau_k = 0 $ would be reached.
However, such limit cannot be achieved in practical implementations and a
stopping condition is needed for the algorithm. By the same assumptions as in
the original paper \citep{Zuev2013} and without loss of generality, the
objective function $\mathcal{H}(\boldsymbol\phi)$ is assumed to be non-negative.
Similarly, let $\delta_k$ denote the Coefficient of Variation (COV) of the
sample $\mathcal{H}(\boldsymbol\phi_1^{(k)}), \ldots,
\mathcal{H}(\boldsymbol\phi_N^{(k)})$, \ie
\begin{align}
\delta_k = \frac{\sqrt{\frac{1}{N} \sum_{i = 1}^N \left( \mathcal{H}\left(\boldsymbol\phi_i^{(k)}
\right) - \frac{1}{N} \sum_{j = 1}^N \mathcal{H} \left(\boldsymbol\phi_j^{(k)} \right) \right)^2
}}{ \frac{1}{N} \sum_{j = 1}^N \mathcal{H} \left(\boldsymbol\phi_j^{(k)} \right)}. 
\end{align}
Therefore, $\delta_k$ is used as a measure of the sensitivity of the objective
function to the hyper-parameters in the domain $\Phi^*_{\tau_k}$. If the samples
are all located in $\Phi^*$ then their COV will be zero, since $\forall \, j \,$
$\mathcal{H}(\boldsymbol\phi_j^{(k)}) = \min_{\boldsymbol\phi \in \Phi^*} \,
\mathcal{H}(\boldsymbol\phi)$. As the progression of the intermediate
distributions advances with $k$, it is expected that $\delta_k \rightarrow 0$.
As a consequence, a criteria to stop the annealing sequence is needed, and the
algorithm will stop when the following condition is attained
\begin{align}
\delta_k < \alpha \, \delta_0 = \delta_{\text{target}},\label{eq:stopping_rule}
\end{align}
where $\alpha$ is assumed to be 0.10 in practical implementations to prevent the
algorithm to generate redundant annealing levels in the last steps of the
procedure. \added{Note that the stopping criterion \eqref{eq:stopping_rule} is
used to drive the simulated annealing temperature towards the absolute zero.
However, if the aim is not localising modes as in stochastic optimisation, and a
more traditional oriented sampling is required, the algorithm could be truncated
in a temperature value of 1. This adds an additional layer of flexibility to the
algorithm which other stochastic-search approaches do not share.}

\subsection{Parallel implementation and guarding against rejection} \label{sec:parallel_delayed}

As found in our earliest experiments, AIMS-OPT with the global acceptance rule
as in Algorithm \ref{alg:global_accep} might degenerate quickly in higher
dimensions since the starting of the chain comes from the highest normalised
weighted sample and a transition might take too long to be performed, resulting
in high rejection rates. Furthermore, information from the markers is lost since
they do not provide good transition neighbourhoods and the ability to create new
samples for the next annealing level is maimed. This aside, AIMS-OPT can become
computationally expensive when the number of samples increases. To cope with
these limitations we propose to incorporate the Transitional Markov Chain Monte
Carlo (TMCMC) and the Delayed Rejection methods into the AIMS-OPT framework.
This extension not only enhances the mixing properties of the sampler, \ie
improve acceptance rates, but also  provides a computational framework in which
parallel Markov chains can be sampled from the intermediate distributions
$p_k(\boldsymbol\phi)$ of the length-scale hyper-parameters.

The idea to enable parallelisation comes from the TMCMC algorithm
\citep[see][for further details]{Ching2007}. In the framework of Algorithm
\ref{alg:annealing_level}, every marker from the annealing level $k-1$ is a
starting point for a Markov chain. This produces not only specialised chains
which are likely to explore the marker's neighbourhood on the sample space, but
also allows an assessment of which markers will generate a better chain. The
normalised weights $\overline{\omega}_j^{(k)}$ will dictate how deeply a chain
will evolve starting from its marker $\boldsymbol\phi_j^{(k-1)}$. Consequently,
the number of samples in each chain will be set with probability proportional to
the normalised weight, a direct result from the TMCMC algorithm.

In order to guard against high rejection rates, and therefore degeneracy on the
sampling scheme, we propose to generate an additional candidate if the first one
is rejected as in Delayed Rejection Algorithms \citep{Mira2001a}. Let
$S_1(\cdot| \cdot)$, $S_2(\cdot | \cdot, \cdot)$ be a one step and two steps
proposal density distributions respectively; $\pi(\cdot)$ the target
distribution of the Markov chain and $a_1(\cdot, \cdot)$ the probability of
accepting a transition in one step.  Then, the probability of accepting a
transition in two steps, denoted by $a_2(\cdot, \cdot)$, is
\begin{align}
a_2(\boldsymbol\phi_0, \boldsymbol\phi_2) = \min \left\lbrace 1, \frac{\pi(\boldsymbol\phi_2) \,
S_1( \boldsymbol\phi_1 | \boldsymbol\phi_2) \, S_2( \boldsymbol\phi_0 | \boldsymbol\phi_2,
\boldsymbol\phi_1) \, (1-a_1(\boldsymbol\phi_2, \boldsymbol\phi_1)) }{ \pi(\boldsymbol\phi_0) \,
S_1( \boldsymbol\phi_1 | \boldsymbol\phi_0) \, S_2( \boldsymbol\phi_2 | \boldsymbol\phi_0,
\boldsymbol\phi_1) \, (1-a_1(\boldsymbol\phi_0, \boldsymbol\phi_1)) } \right\rbrace,
\end{align}
where $\boldsymbol\phi_0$ denotes the starting point, $\boldsymbol\phi_1$ the
rejected candidate and $\boldsymbol\phi_2$ the second stage candidate. In our
context, the target distribution $\pi(\cdot)$ is each annealing level
$p_k(\cdot)$ density distribution, the one step proposal distribution $S_1$ is
the independent approximation in equation \eqref{eq:approx_comp} and the 
one-step acceptance probability is the global acceptance probability in
\eqref{eq:global_accep}. The two-step proposal density $S_2$ can be chosen from
several alternatives. In this work we use a symmetric distribution centred at
the starting point $\boldsymbol\phi_0$, since it can be seen as a back-guard
against $S_1$ being a deficient independent sampler \citep[see][for a detailed
discussion]{Zuev2011}. Therefore, the previous equation can be rewritten in
compact form as
\begin{align}
\alpha_{k,2}(\boldsymbol\phi_0, \boldsymbol\phi_2) = \min \left\lbrace 1,
\frac{p_k(\boldsymbol\phi_2) \, (1 - \alpha_k^g(\boldsymbol\phi_1 | \boldsymbol\phi_2) )
}{p_k(\boldsymbol\phi_0) \, (1 - \alpha_k^g (\boldsymbol\phi_1 | \boldsymbol\phi_0 ) ) }
\right\rbrace, \label{eq:twostepaccep}
\end{align}
where $\alpha_k^g(\cdot|\cdot)$ is defined as in equation
\eqref{eq:global_accep}. The fact that $S_2$ is a symmetric distribution centred
in the starting point $\boldsymbol\phi_0$ has been used, \ie
$S_2(\boldsymbol\phi_2 | \boldsymbol\phi_0, \boldsymbol\phi_1 ) =
g(\boldsymbol\phi_2 | \boldsymbol\phi_0) = g(\boldsymbol\phi_0 |
\boldsymbol\phi_2) = S_2(\boldsymbol\phi_0 | \boldsymbol\phi_2,
\boldsymbol\phi_1 )$, where $g(\cdot|\cdot)$ denotes such symmetric proposal
density. By performing the second stage proposal,  the stationary condition of
$p_k(\cdot)$ is maintained as stated in the following proposition.
\begin{proposition}
AIMS-OPT coupled with delayed rejection in two stages leaves the target
distribution $p_k(\cdot)$ invariant at each annealing level.
\end{proposition}
\begin{proof}		
See \ref{appendix} for a proof using a general transition distribution
$S_2(\cdot | \cdot, \cdot)$.
\end{proof}

From the above discussion, the proposed scheme provides a fail-safe against any
possible mismatch of the approximation done with \eqref{eq:approx_comp}.
Additionally, the results presented in this paper correspond to the second step
candidate being a Gaussian random variable, $\boldsymbol\xi \sim
\mathcal{N}(\boldsymbol\phi_i^{(k)} | c_0 \Sigma_k)$. The ideas to accept a
global transition after having accepted a local proposition can be summarised in
Algorithm \ref{alg:global_accep_delayed}.

\begin{algorithm}[H]
\eIf{$\boldsymbol\xi$ was accepted as local candidate}{
Accept $\boldsymbol\xi$ as a global transition with probability
\vspace{-2.5mm}
\begin{align}
\alpha^{g}_k\left( \boldsymbol\xi \left\vert \boldsymbol\phi^{(k)}_i \right. \right) \;
\end{align}
\vspace{-5.5mm}
}{
Generate a second candidate $\boldsymbol\xi_2$ from 
\begin{align}
\boldsymbol\xi_2 \sim \mathcal{N}(\boldsymbol\phi_i^{(k)} | c_0 \Sigma_k)
\end{align}
\eIf{$\boldsymbol\xi_2$ is accepted with probability $\alpha_{k,2}(\boldsymbol\phi_{i}^{(k)},
\boldsymbol\xi_2) $ computed as in equation \eqref{eq:twostepaccep} }{
\vspace{-2.5mm}
\begin{align}
\boldsymbol\phi_{i+1}^{(k)} = \boldsymbol\xi_2 \;
\end{align}
\vspace{-5.5mm}
}{
\vspace{-2.5mm}
\begin{align}
\boldsymbol\phi_{i+1}^{(k)} = \boldsymbol\phi_{i}^{(k)} \;
\end{align}
\vspace{-5.5mm} }
}
\caption{Global acceptance using delayed rejection}\label{alg:global_accep_delayed}
\end{algorithm}

%------------------------------------------------

\section{Implementation Aspects} \label{sec:implementation}

The computational complexity of the posterior distribution in equation
\eqref{eq:integrated_post} is governed by the inverse of the covariance matrix
$K$ as it scales with the number of training runs $N$. Several solutions have
been developed in the literature, such as computation of inverse products of the
form $K^{-1} u$, with $u \in \mathbb{R}^N$, by means of Cholesky factors or
Spectral Decomposition \citep[see][for efficient implementations]{Golub1996} to
preserve numerical stability in the matrix operations \citep[see][]{Gibbs1998}.
Nonetheless, numerical stability is not likely to be achieved if the training
runs are very limited, or if the sampling scheme for such training runs cannot
lead to stable covariance matrices, as depicted in Figure \ref{fig:numstab}.

\begin{figure}[H]
\centering 
\subfloat[Numerically unstable surface]{\includegraphics[draft=false,width=.32\linewidth]{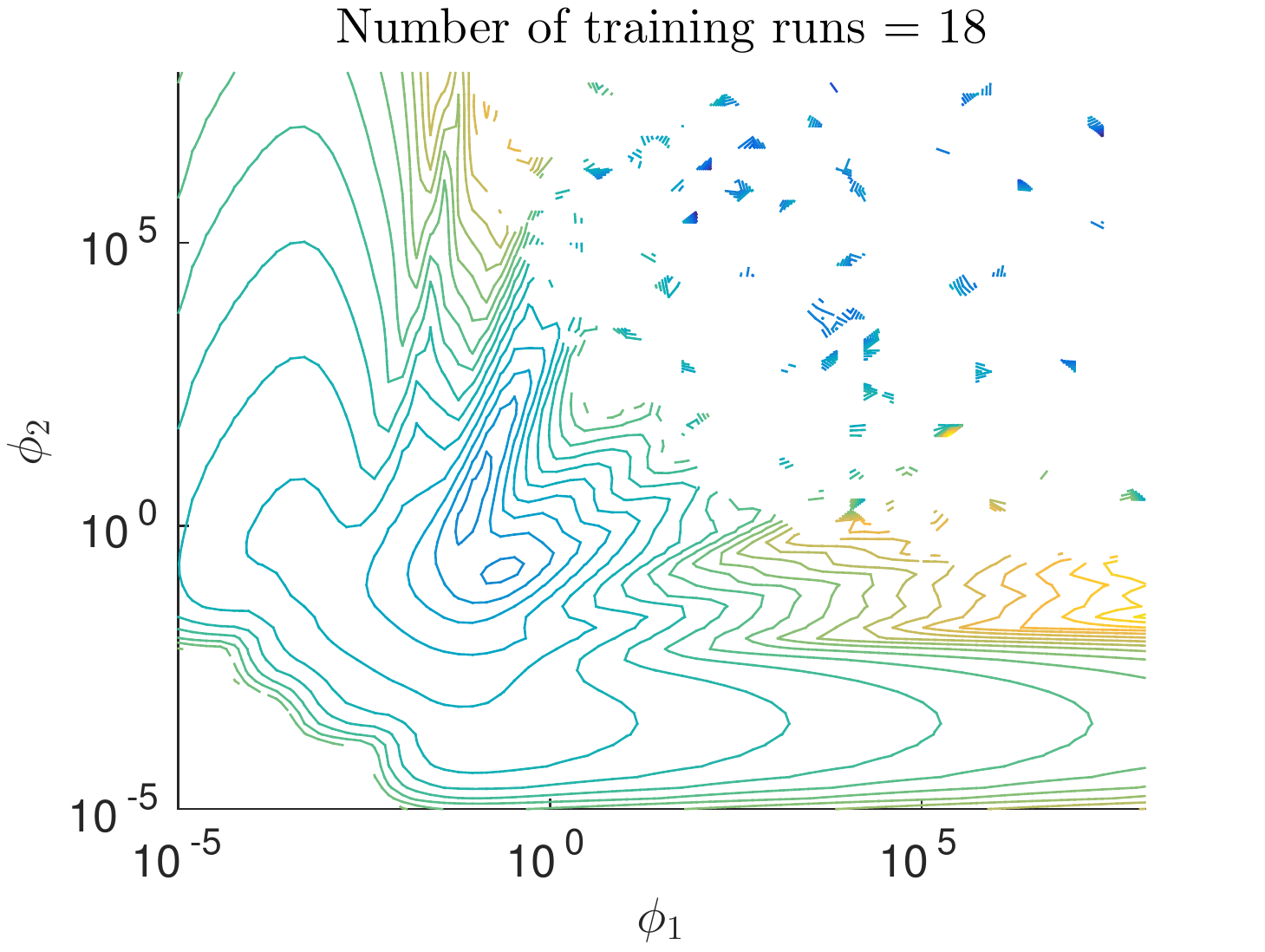}\label{subfig:numstab01}}
\qquad
\subfloat[Numerically stable surface]{\includegraphics[draft=false, width=.32\linewidth]{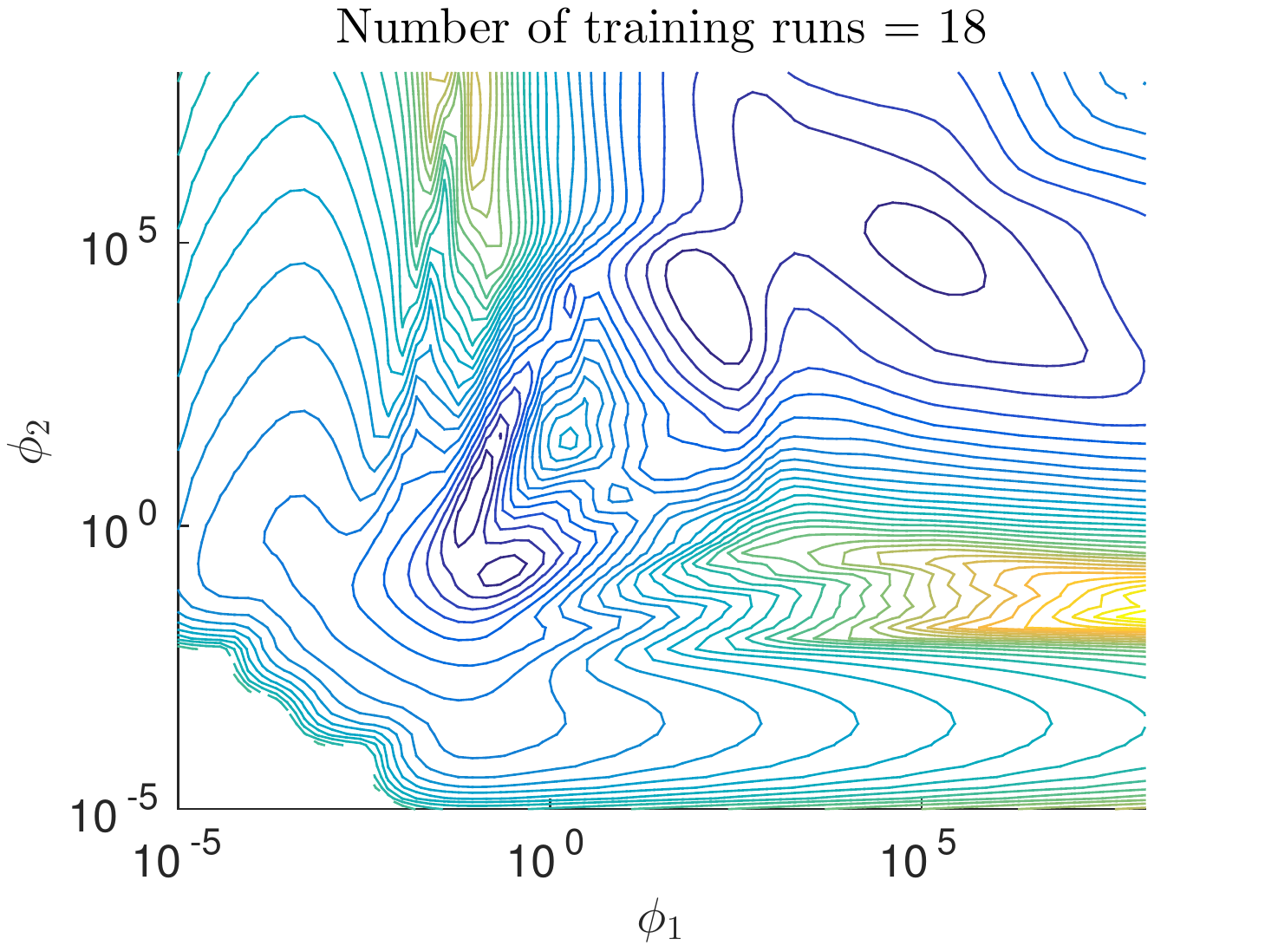}\label{subfig:numstab02}}
\caption{Projection of the negative log-posterior curves in the two dimensional length-scale space. Adding
the nugget $\phi_{\delta}$ results in a numerically stable surface.}
\label{fig:numstab}
\end{figure}
	
To overcome this practical deficiency, a correction term in the covariance
matrix can be added in order to preserve diagonal dominancy, that is, we add a
{\em nugget} hyper-parameter $\boldsymbol\phi_\delta$ to the covariance such
that
\begin{align}
K_\delta = K + \boldsymbol\phi_\delta \, I, \label{eq:nugget_cov}
\end{align}
is positive definite. Doing so results in the stochastic simulator 
\begin{align}
y_i = \eta(\vec{x}_i) + \sigma^2 \, \boldsymbol\phi_\delta.
\end{align}

Note that the interpolating quality of the Gaussian process is lost, however,
the term $\sigma^2 \, \boldsymbol\phi_\delta$ accounts for the variability of
the simulator that cannot be explained by the emulator given the original
assumptions (adequacy of the covariance function, for example). The nugget can
also provide further quantification of model uncertainty in the inference
process as it provides an alternative to smoothing an already complex surface.
As it is also noticed by \cite{Andrianakis2012} and \cite{Ranjan2011}, the
quality of the emulator changes with the inclusion of the nugget, since it
modifies the objective function itself by introducing new modes in the landscape
of the posterior distribution. \replaced{The configuration reflected by new
modes in these cases might correspond to emulators with no local dependencies
and an overall simple trend, defined from the basis functions and regression
hyper-parameters $\boldsymbol\beta$. Therefore, if a Gaussian process with no
local dependencies, \eg with its mode farther away from the origin in the
length-scale space,}{ Note that if such modified landscape} is assessed as not
appropriate for the model, a regularisation term can be added in the
optimisation formulation \added{as in} \citep{Andrianakis2012}. \added{This
corresponds to precautions for the inclusion of the nugget and can be seen as
elicited prior beliefs on the Bayesian formulation}. However, by using a 
multi-modal sampler for stochastic optimisation as the one proposed, a robust
emulator capable of mixing various possibilities can be provided. This results
in an emulator that is able to cope with violations to the  modelling
assumptions originated by working with a limited amount of training runs.

We incorporate the nugget term $\boldsymbol\phi_\delta$ as a hyper-parameter of
the correlation function in the Bayesian inference process. As suggested by
\cite{Ranjan2011} a uniform prior distribution $U(10^{-12}, 1)$ for such
parameter is considered. The effect of the bounds is twofold. First, the lower
bound is used to guarantee stability in the covariance matrix. Second, the upper
bound is used to force the numerical noise of the simulator to be smaller than
the signal noise of the emulator itself. Note that this last assumption can be
omitted if the problem requires it. By considering the correlation matrix as in
equation \eqref{eq:nugget_cov}, this yields
\begin{align}
\Sigma_\delta &= \sigma^2 \, K_\delta ,
\end{align}
where $K_\delta$ denotes the corrected correlation matrix and $\Sigma_\delta$
has been used to denote the covariance matrix of the Gaussian process. By doing
so it is clear that previous considerations regarding $\sigma^2$, such as the
ability of marginalising it as a nuisance parameter and the use of a 
non-informative prior remain unchanged \citep{DeOliveira2007}.

%------------------------------------------------

\section{Numerical Experiments} \label{sec:numerical}

To illustrate the robustness of estimating the hyper-parameters of a Gaussian
process using the parallel AIMS-OPT framework, three test cases have been
selected. The first two are common examples that can be found in the literature.
The first is known as the Branin function and has been modified to resemble
usual properties of engineering applications \citep[]{Forrester2008}. The second
one \citep{Bastos2009} has been used as a two dimensional function with a
challenging complexity for emulating purposes. The third example presented in
this section comes from a real dataset also presented in \cite{Bastos2009}. In
all the examples it is assumed that $h(\vec{x}) = (1, x_1, \ldots, x_p)^T$.
Regarding the nugget, a sigmoid transformation has been performed in order to
sample from a Gaussian distribution. Namely, we sample an auxiliary $z_\delta$
as part of the multivariate Gaussian in  \eqref{eq:gauss_local}, and compute the
nugget as
\begin{align}
\theta_\delta = \frac{1-l_b}{1+\exp(-z_\delta)} + l_b,
\end{align}
where $l_b$ is the lower bound for the nugget, which is set equal to $10^{-12}$.
Additionally, the uniform meta-prior distribution of equation
\eqref{eq:meta_prior} has been considered in a practical support of the 
length-scale parameters in the logarithmic space, namely a uniform distribution
with support in $[-7,7]$. For the nugget, a truncated beta distribution with
parameters $\alpha = \beta = 0.5$ has been considered since it corresponds to a
non informative meta-prior in the interval $[l_b, \, 1]$. Here the prefix {\em
meta} has been used to refer to the algorithm's prior distribution and to set a
clear distinction from the prior used in the modelling assumptions in equation
\eqref{eq:prior}.

The code has been implemented in MATLAB and all examples have been run in a
GNU/Linux machine with an Intel i5 processor with 8 Gb of RAM. For the purpose
of reproducibility, the code used to generate the examples \deleted{in this
paper} is available for download at
\url{https://github.com/agarbuno/paims_codes}.

\subsection{Branin Function} \label{subsec:branin}

The version of the Branin function used in this paper is a modification made by
\cite{Forrester2008} for the purpose of Kriging prediction in engineering
applications. It is a rescaled version of the original in order to bound the
inputs to the rectangle $[0,1]\times [0,1]$, with an additional term that
modifies its landscape to include a global optimum. Namely,
\begin{align}
f({\bf x}) = \left( \overline{x}_2 - \frac{5.1}{4\pi^2} \overline{x}_1^2 + \frac{5}{\pi}
\overline{x}_1 - 6 	\right)^2 + 10 \left[ \left(1- \frac{1}{8\pi}	\right)
\cos(\overline{x}_1)+ 1	\right] + 5 \overline{x}_1 ,
\end{align}
where $\overline{x}_1 = 15\,x_1 -5$ and $\overline{x}_2 = 15\, x_2$.

For this case, a sample of 18 design points were chosen with a Latin hypercube
sampling scheme. The resulting log-posterior function possesses 4 different
modes in its landscape (see Figure \ref{subfig:branin}) leading to 4 possible
configurations of the correlation function. Thus, the impact of the training
runs used to construct the emulator is evident. Among these modes, 4 different
types of emulators can be distinguished: an emulator with high sensitivity to
changes in input $x_1$ (mode A in Figure \ref{subfig:branin}); an emulator with
rapid changes in $x_2$ for the correlation structure of the training runs (mode
B); a limiting case where dimension $x_2$ is disregarded in the correlation
function, due to a high value in $\boldsymbol\phi_2$ (mode C); or a second
limiting emulator which approximates a Bayesian linear regression model (mode D)
\citep[see][for a detailed discussion]{Andrianakis2012}.
\begin{figure}[H]
\centering 
\subfloat[Level curves and modes]{\includegraphics[draft=false,width=.32\linewidth]{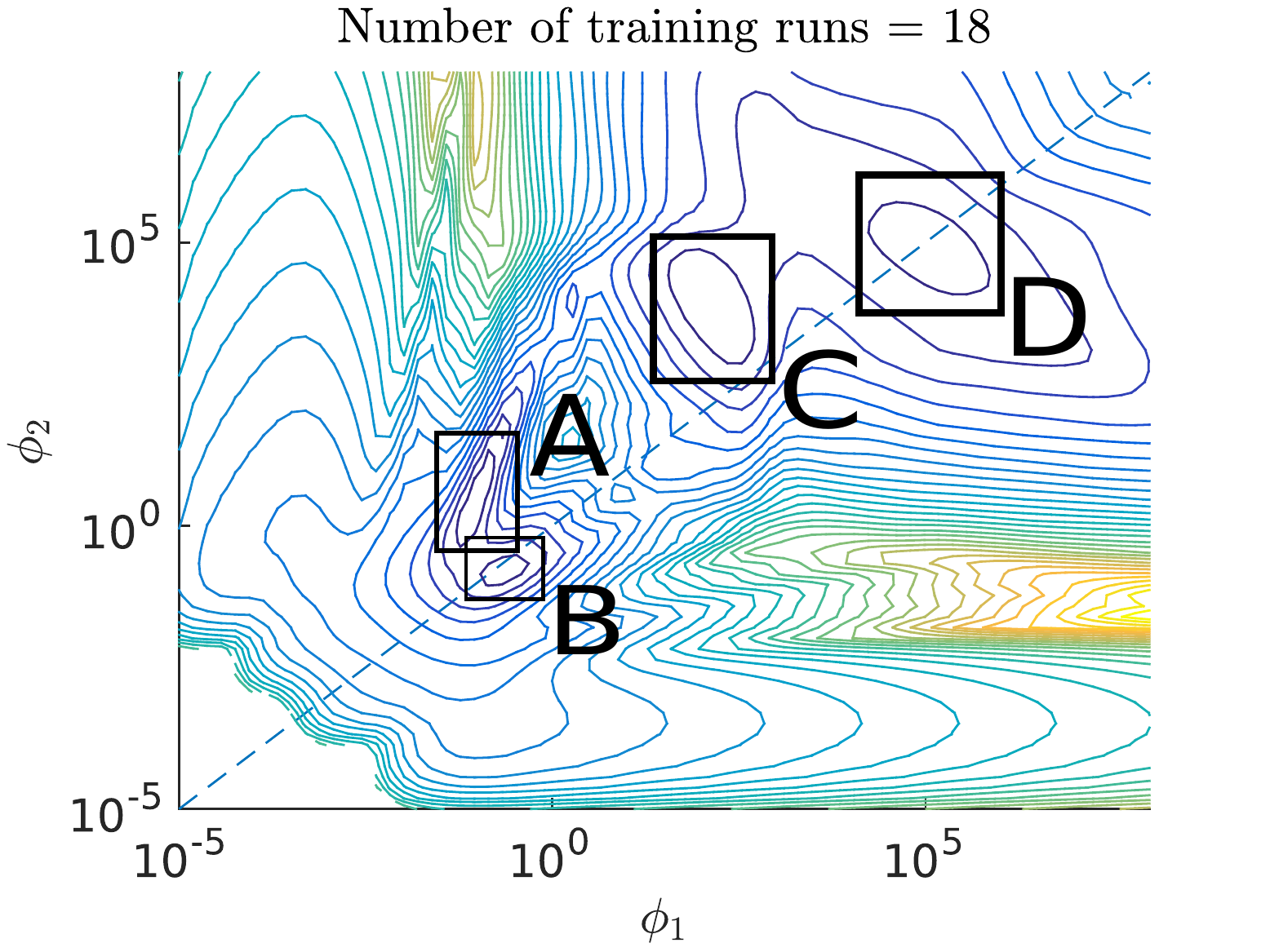}\label{subfig:branin}}
\subfloat[Parallel AIMS-OPT sample]{\includegraphics[draft=false,width=.32\linewidth]{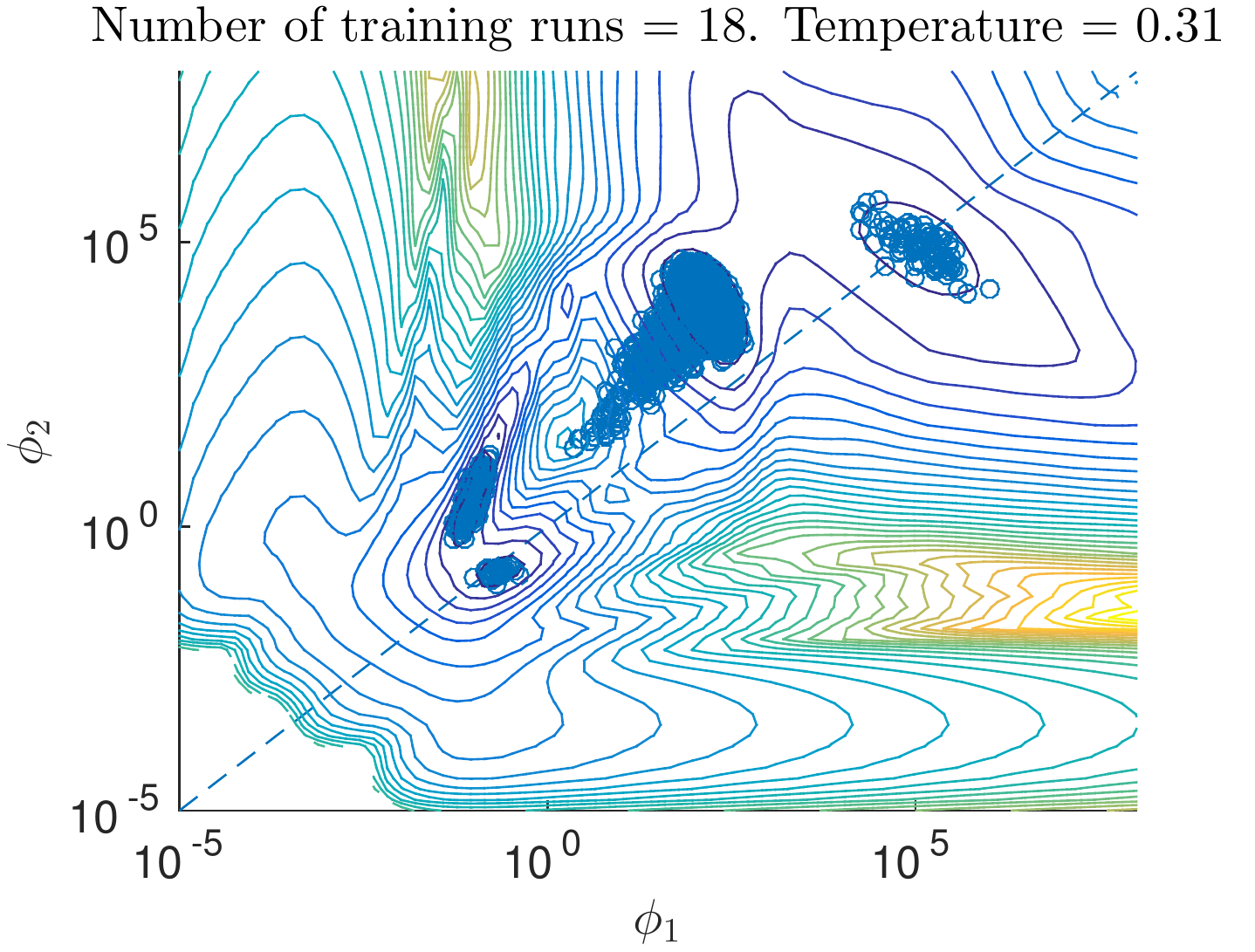}\label{subfig:branin_samples}}
\subfloat[Residual plots]{\includegraphics[draft=false,width=.32\linewidth]{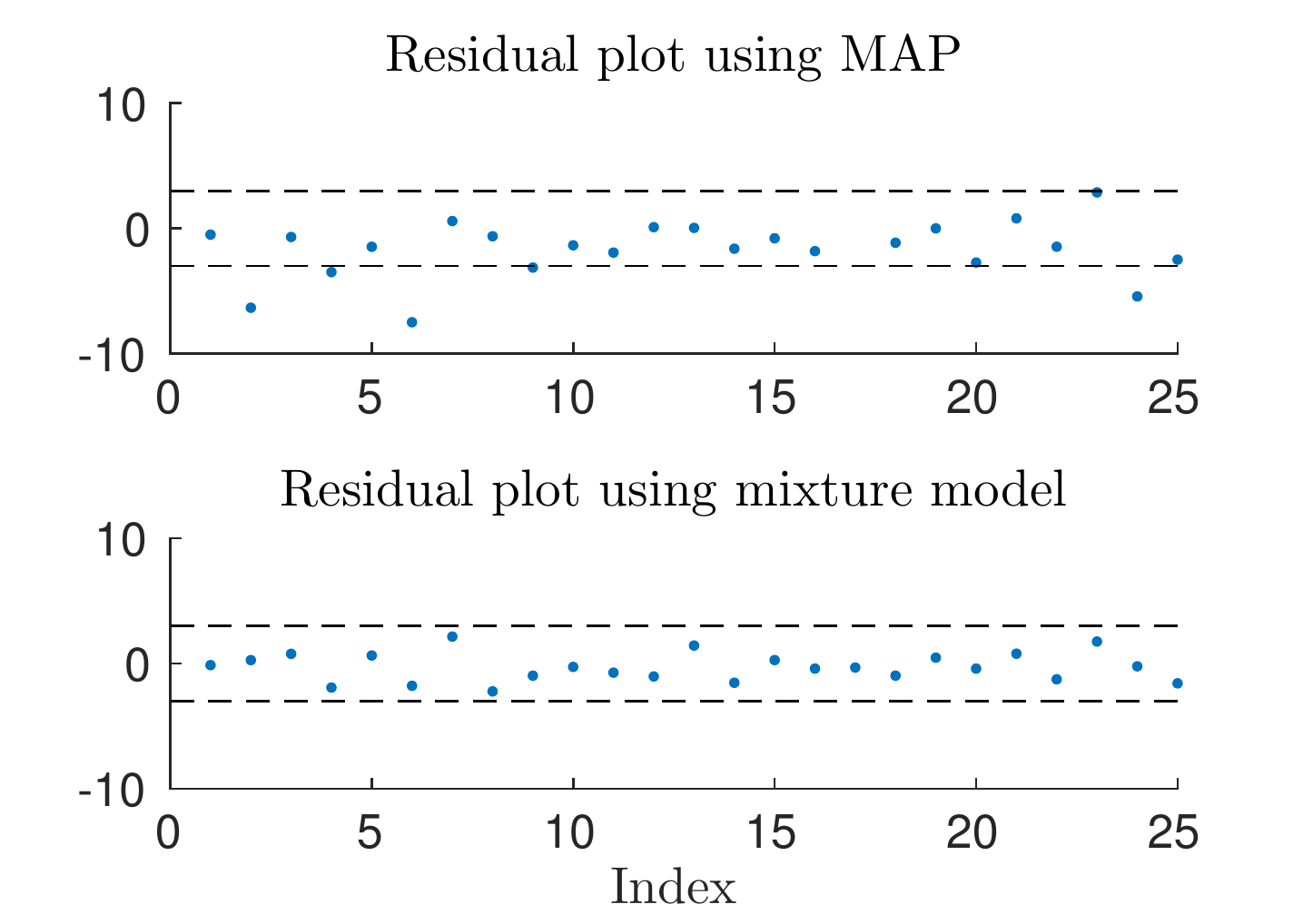}\label{subfig:branin_res}}
\caption{Projection of the negative log-posterior curves in the two dimensional
length-scale space for the Branin simulator. The minimum possible value of
$10^{-12}$ for the nugget $\boldsymbol\phi_\delta$ has been used for such
projection. The reference diagonal helps visualise the regions where the length
scales favour one dimension over the other.}
\label{fig:loglik_branin}
\end{figure}

For this example, two thousand samples were generated in each annealing level.
The parallel AIMS-OPT algorithm generated 7 annealing levels to produce the
samples in Figure \ref{subfig:branin_samples}. The RMSE of the MAP model is
7.068 whereas the RMSE of the mixture is 15.099 which is an indication that in
terms of brute prediction, the mixture model could be improved by taking more
samples. Figure \ref{subfig:branin_res} depicts the standardised residuals from
both the MAP approach (top) and the mixture model (bottom) using equations
\eqref{eq:mix_mean} and \eqref{eq:mix_cov} with uniform weights in the
sample. The standardised residuals are defined as
\begin{align}
r(\vec{x}) = \frac{y - \mu(\vec{x}) }{\sqrt{\sigma^2(\vec{x})}}, \label{eq:residuals}
\end{align}
where $y$ is the output for configuration $\vec{x}$, $\mu(\vec{x}) = E [y |
\vec{x}, \mathcal{D}]$ and $\sigma^2(\vec{x}) = \text{var} (y | \vec{x},
\mathcal{D} )$, the posterior mean and variance for configuration $\vec{x}$
\citep[see][\added{for a more detailed discussion on diagnostics}]{Bastos2009}.
By marginalising the hyper-parameters it is clear that our estimation is a more
robust in terms of error prediction. Even with such limited amount of
information the residuals suggest that the uncertainty is being incorporated
appropriately in the marginalised predictive posterior distribution in equation
\eqref{eq:pred_post}. The standardised residuals are inside the \replaced{95\%
confidence bands, assuming approximate normality}{ bands}, though not too close
to 0. This \replaced{is an indicator that although greater variability is
expected, excessively large variances are avoided. This is done by means of the
integrated predictive distribution and the use of the proposed sampler to build
a mixture of emulators leaving the predicted errors inside appropriate bounds}{
avoids an excessively large variance for the predictive distribution}.

\subsection{2D Model}

This function has already been used as an example for emulation purposes and can
be found in GEM-SA software web page
(\url{http://ctcd.group.shef.ac.uk/gem.html}). Even though it is a two
dimensional problem it also serves as a good illustration of the importance of
estimating the hyper-parameters of a Gaussian process with a multi-modal
sampler. The mathematical expression for this simulator is
\begin{align}
f({\bf x}) = \left[1-\exp\left(-\frac{0.5}{x_2}\right)\right] \, \left(\frac{ 2300 x_1^3 + 1900
x_1^2 + 2092 x_1 + 60}{100 x_1^2 + 500 x_1^2 + 4 x_1 + 20} \right). 
\end{align}

As in the previous case, the training runs and the modelling assumptions fail to
summarise the uncertainty in a uni-modal posterior distribution. The design
points where selected using a Latin hypercube in the rectangle
$[0,1]\times[0,1]$. It can be seen from Figure \ref{subfig:bastos_1} that the
modes are separated by a wide valley of low posterior probability, which can
become an overwhelming task for traditional MCMC samplers. The proposed sampler
is able to cope with all local and global spread dynamics present in the
neighbourhoods of the modes it encounters, as shown in Figure
\ref{subfig:bastos_2}.

\begin{figure}[H]
\centering 
\subfloat[Level curves with reference prior]{\includegraphics[draft=false,width=.32\linewidth]{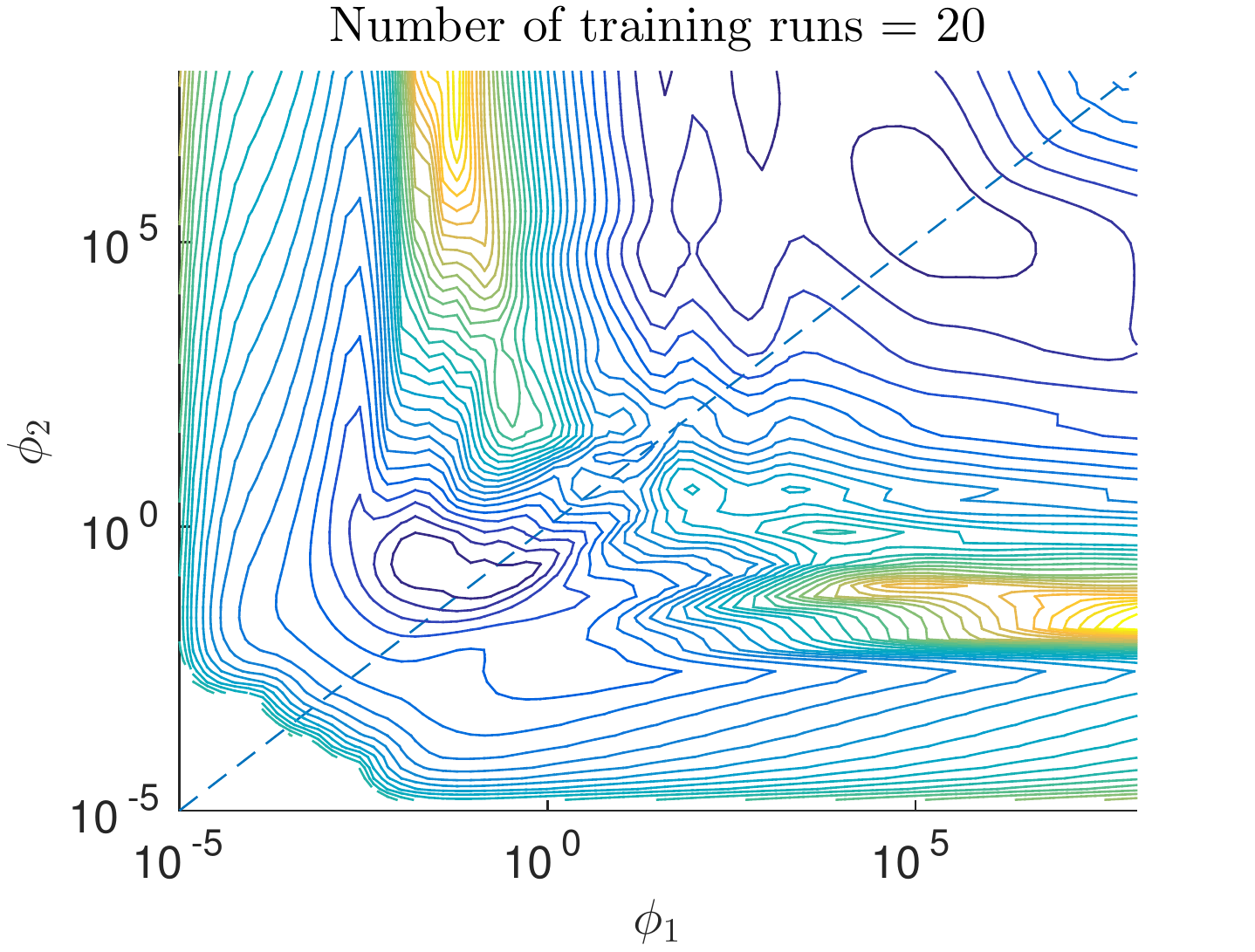}\label{subfig:bastos_1}}
\qquad
\subfloat[Parallel AIMS-OPT sample]{\includegraphics[draft=false,width=.32\linewidth]{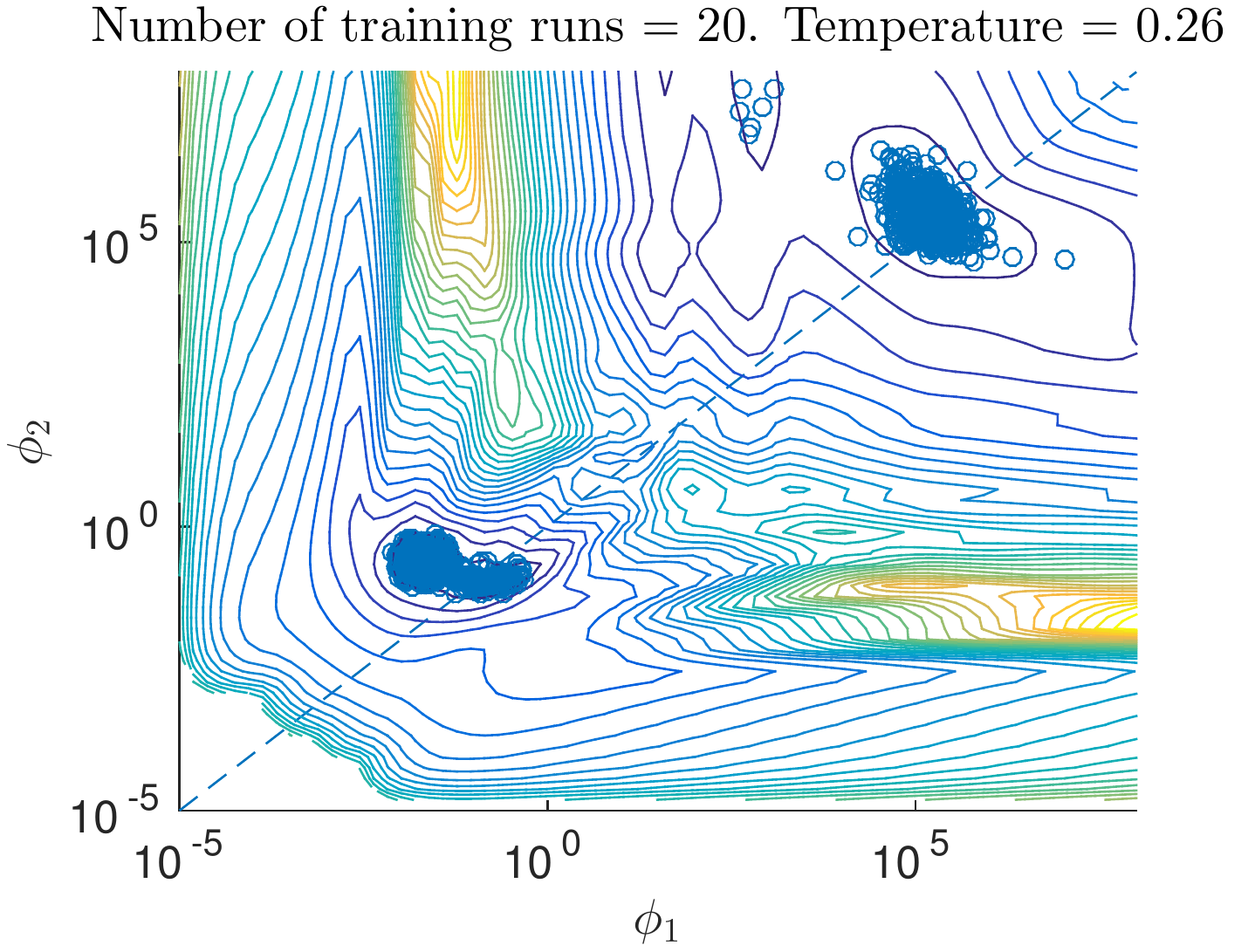}\label{subfig:bastos_2}} \\
\subfloat[Level curves with uniform prior]{\includegraphics[draft=false,width=.32\linewidth]{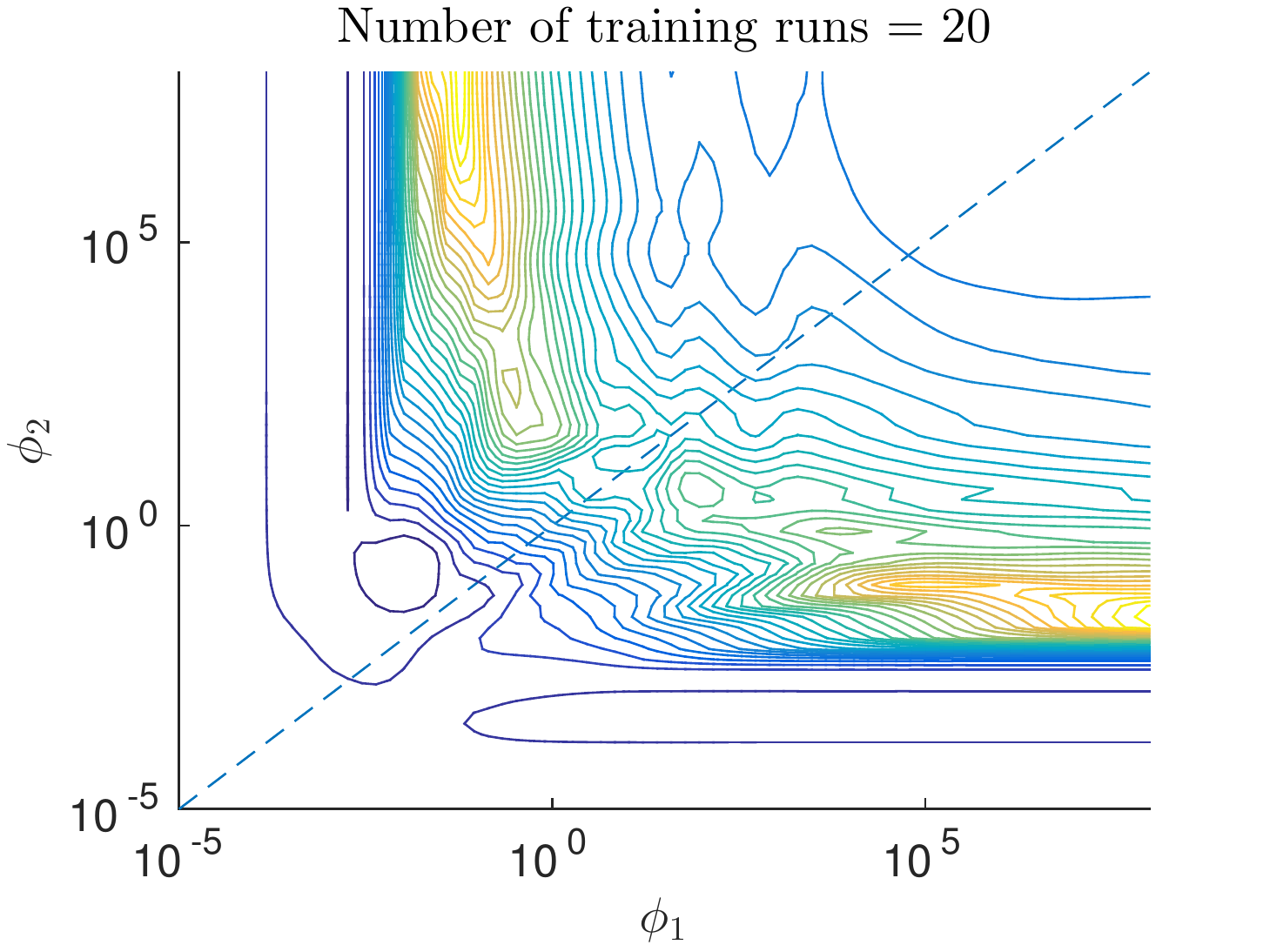}\label{subfig:bastos_3}} \qquad
\subfloat[Residuals plot]{\includegraphics[draft=false,width=.32\linewidth]{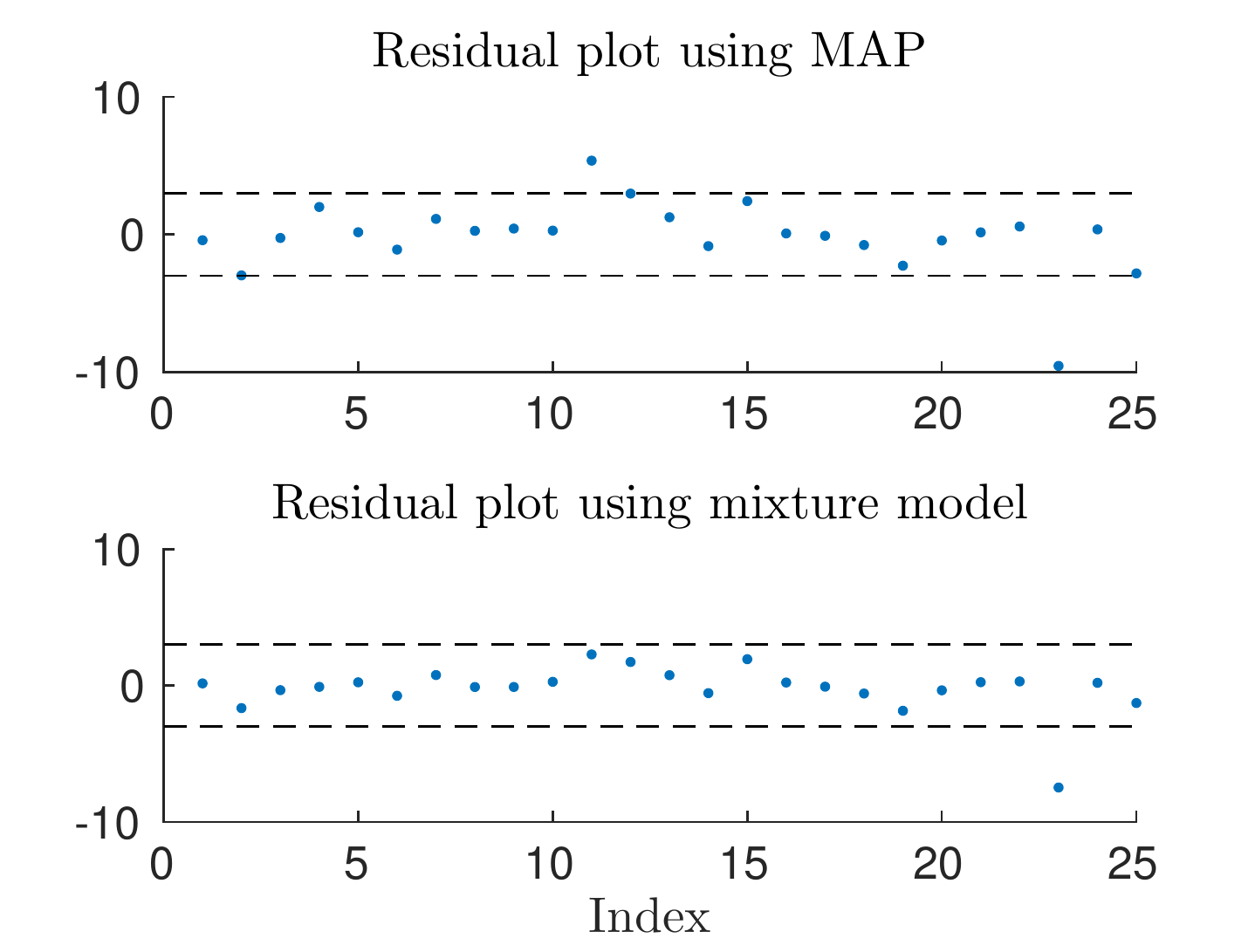}\label{subfig:bastos_res}
} 
\caption{Projection of the negative log-posterior curves in the two dimensional
length-scale space for the 2D Model simulator. The minimum possible value of
$10^{-12}$ for the nugget $\boldsymbol\phi_\delta$ has been used for such
projection. The reference diagonal helps visualise the regions where the length
scales favour one dimension over the other.}
\label{fig:loglik_bastos}
\end{figure}

Depicted in Figures \ref{subfig:bastos_1} and \ref{subfig:bastos_3} the use of
the reference prior in the posterior distribution removes probability mass from
the neighbourhood around the origin. This validates the use of the reference
prior to cut out regions from the \deleted{sample } space \added{of 
hyper-parameters} for the sampling and exploit the most information contained
in the data available, namely, the training runs $\mathcal{D}$. As in the
previous example, two thousand samples were generated in each annealing level.
The parallel AIMS-OPT algorithm generated 7 annealing levels to produce the
samples in Figure \ref{subfig:bastos_2}. In terms of prediction accuracy, we now
obtain that the RMSE is 1.356 for the MAP estimate and 1.345 for the mixture
model. While as for the residuals, we can see from Figure
\ref{subfig:bastos_res} that the mixture model \replaced{allows for a more
robust prediction of the error, by means of increasing the variability in
particular locations. This can be seen as the standardised residuals are
concentrated within the 95\% confidence bands of an approximate assumed
normality}{ consistently narrows the spread of the prediction}, resulting in a
more robust estimation of the error \replaced{by the use of a mixture model}{
predictions}. \added{This motivates the use of multi-modal density samplers in
the context of optimisation, where if a single candidate is provided the overall
error prediction of the emulator might be biased towards more concentrated
predictions around the mean estimation.}

\subsection{Nilson-Kuusk Model}

This simulator is built from the Nilson-Kuusk model for the reflectance for
\added{a} homogeneous plant canopy. Such model is a five dimensional simulator
whose inputs are the solar zenith angle, the leaf area index, relative leaf
size, the Markov clumping parameter and a model parameter $\lambda$
\citep[see][for further details on the model itself and the meaning of the
inputs and outputs]{Nilson1989}. For the analysis presented in this paper a
single output emulator is assumed and the set of the inputs have been rescaled
to fit the hyper-rectangle $[0,1]^5$ on a five dimensional space as in
\cite{Bastos2009}.

As in the previous test cases, the design points were chosen by Latin hypercube
designs (100 for this case). In this example, the dimension of the problem makes
it impossible to plot the level curves of the posterior distribution for the
length scale hyper-parameters \added{to visualize potential multiple modes}.
However, \replaced{the samples can be visualized by means of a box-plot as shown
in Figure \ref{fig:kuusk_plot}, where the red line denotes the median, the edges
of the box the 25\textsuperscript{th} and 75\textsuperscript{th} percentiles,
and the whiskers cover the most extreme cases. The samples are obtained after
completing 10 levels of the parallel AIMS-OPT algorithm. The box-plots of the
approximate optimal solutions}{ the plots of the estimated densities computed
from the sample obtained after 11 levels of the parallel AIMS-OPT algorithm}
strongly suggest that the samples come from a multi-modal posterior
distribution. \replaced{This can be seen from the location of the edges of the
boxes and the median for any given input. The last input possesses a very
limited spread which might denote a high concentration around one mode. Note
that as the magnitude of the length-scale increases, thus reducing the
sensitivity of the simulator to such input, the length-scales are located in
what can be seen as either a plateau or regions of modes with negligible
difference in the posterior density.}{ , even though the samples are
concentrated around one mode.} Additionally, \added{from the range of values
that are covered in log-space,} it can be noted that \replaced{the output of the
simulator appears to be insensitive to changes of the third and fourth input}{
for the third and fourth input of the model the output changes very smoothly}.
Furthermore, a limit-case emulator can be suggested by the \added{box-}plot in
\replaced{Figure \ref{fig:kuusk_plot}}{\ref{subfig:kuusk_phi4}} by considering a
surrogate with no \added{third and} fourth inputs in the model. Notice the
scale\added{s} for such hyper-parameter\replaced{s in logarithmic space}{ and
the mode appearing around $e^{10}$}.
\begin{figure}[H]
\centering 
\includegraphics[draft=false, width=.39\linewidth]{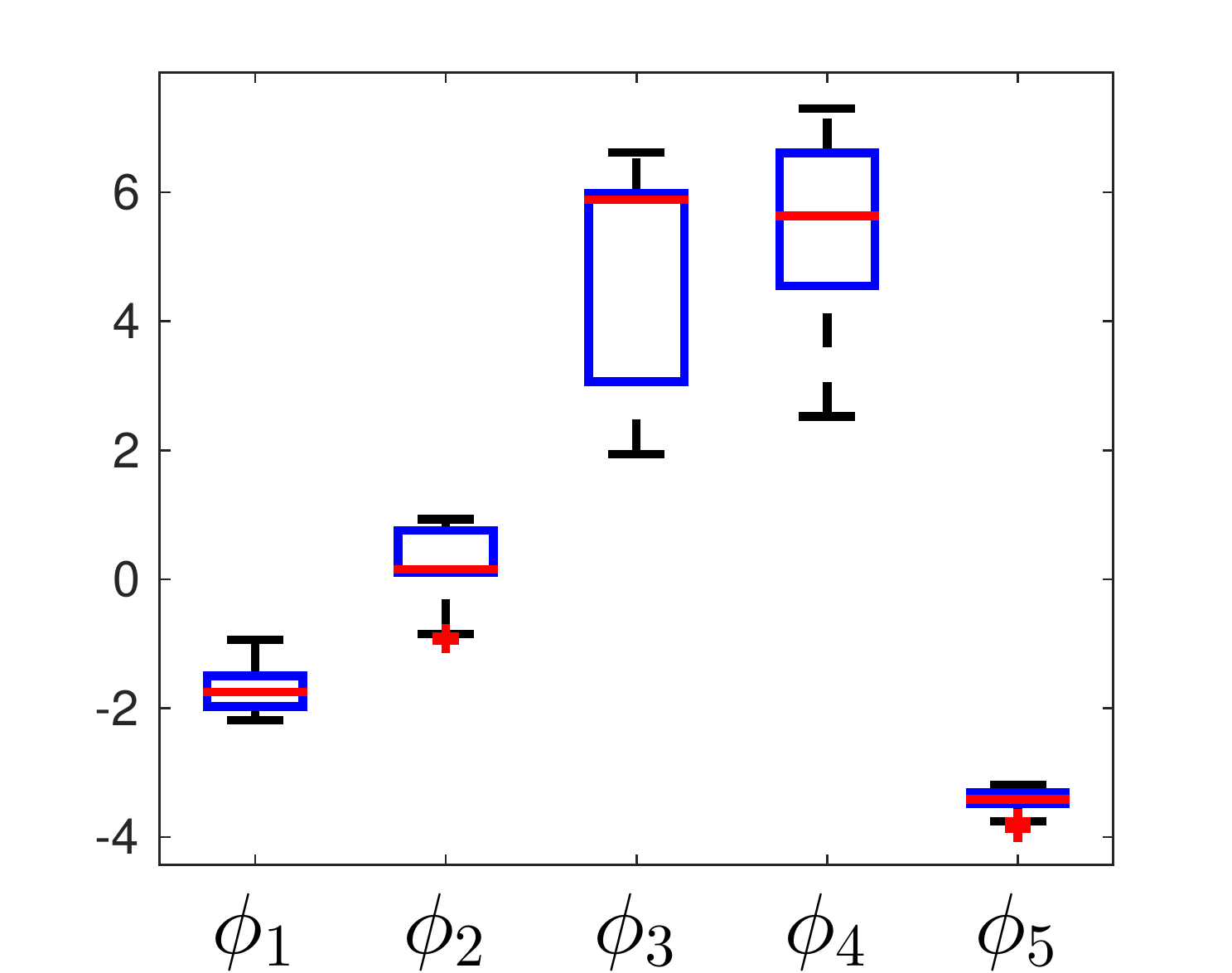}
\caption{\replaced{Box-plots of the sample of length-scales obtained by the
parallel asymptotically independent Markov sampling.}{Estimated densities for
the hyper-parameters and objective function obtained from the Nilson-Kuusk
model.}}
\label{fig:kuusk_plot}
\end{figure}

Due to the larger number of dimensions, five thousand samples were generated for
each annealing level. In this case we have that the RMSE of the MAP estimate is
0.022 while the RMSE of the mixture proposal is 0.021 which is a consequence of
the \replaced{posterior distribution}{ sample} being highly concentrated around
one mode\added{, in a particular set of length-scales ($\phi_1, \phi_2$ and
$\phi_5$) while being less specialised for the less sensitive ones ($\phi_3$ and
$\phi_4$)}. In Figure \ref{fig:kuusk_res} there is evidence that even with such
behaviour the predictive error is improved by narrowing the spread of the
standardised residuals, \added{as before, a consequence of an increased
estimation of the variability in particular locations}. In this case the
residuals cannot all be contained in the \added{approximate normality 95\%}
bands but as noted by \cite{Bastos2009} in their experiments there is strong
evidence that \deleted{in this case} more runs of the simulator are needed to
adequately built a statistical surrogate. \added{Due to the highly concentrated
posterior density around the high sensitive length-scales there seems to be no
apparent gain from using the mixture model. However, it can be noted from Figure
\ref{fig:kuusk_plot} that by acknowledging the variability of the 
hyper-parameters, a better understanding of the sensitivity of the simulator
with respect to the inputs is achieved. An improved and more robust uncertainty
analysis of the simulator can be provided in this case understanding the wide
spread of length-scales for particular dimensions. For instance if screening is
performed, the MAP estimate will fail to summarise the wide posterior density
with respect to $\phi_3$ and $\phi_4$ and this in turn, will provide partial
information. This analysis  cannot be performed solely by maximising the
posterior density. Therefore the proposed method provides additional insight of
the sensitivity of both simulator and emulator.}
\begin{figure}[H] \centering
\includegraphics[draft=false,width=.39\linewidth]{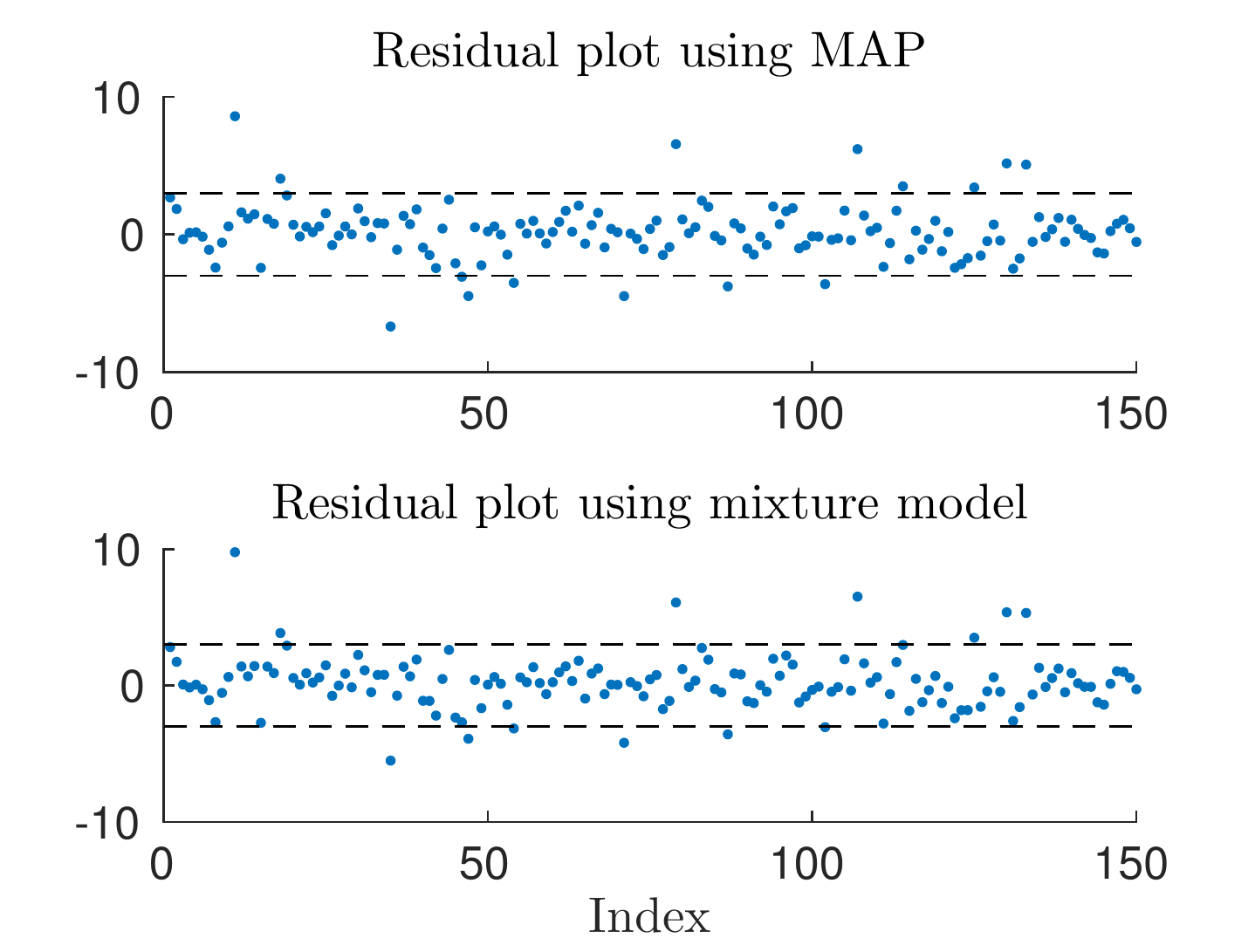}
\caption{Residuals plot for the Nilson-Kuusk simulator}  \label{fig:kuusk_res}
\end{figure}

\section{Conclusions}	\label{sec:conclusions}

This paper proposes to estimate the hyper-parameters of a Gaussian process using
a new sampler based on the Asymptotically Independent Markov Sampling (AIMS)
method. The AIMS-OPT algorithm, used in stochastic optimisation, provides a
robust computation of the MAP estimates of the hyper-parameters. This is done by
providing a set of approximations to the optimal solution instead of a single
approximation as it is so frequently done in the literature. The problem is
approached in a combined effort from the computational, optimisation and
probabilistic perspectives which serve as solid foundations for building
surrogate models for computationally expensive computer codes.

The original AIMS algorithm has been extended to provide an efficient sampler in
computational terms, by means of parallelisation, as well as an effective
sampler with good mixing qualities, by means of both the delayed rejection and
adaptive modification exposed. It has been demonstrated that by using the
parallel AIMS-OPT algorithm it is possible to acknowledge uncertainty in the
structure of the emulator proposed as illustrated in the examples provided.
Structural uncertainty should be taken into account to determine when the
training runs available are sufficient to narrow the posterior distribution of
the hyper-parameters to a uni-modal convex distribution. Even though it has been
proven to be effective in lower and medium dimensional design spaces, research
in high dimensional spaces has been left for future research.

%------------------------------------------------

\section*{Acknowledgements}

The first author gratefully acknowledges the Consejo Nacional de Ciencia y
Tecnolog\'{\i}a (CONACYT) for the award of a scholarship from the Mexican
government.

%------------------------------------------------
\appendix
\section{}\label{appendix}

In this appendix, a proof that using the delayed rejection algorithm in the AIMS
framework leaves the target distribution $p_k(\cdot)$ invariant is provided.

A sufficient condition to prove that indeed $p_k(\cdot)$ is the stationary
distribution for the Markov chain is to prove that the detailed balance
condition is satisfied. Since the first stage approval has been proven to
satisfy the detailed balance condition in \cite{Zuev2013}, it will only be
proved for the second stage sampling.

Let $f_k(\boldsymbol\phi_2| \boldsymbol\phi_0)$ describe the AIMS-OPT delayed
transitions in the $k$-th annealing level from $\boldsymbol\phi_0 \rightarrow
\boldsymbol\phi_2$, with $\boldsymbol\phi_2 \neq \boldsymbol\phi_0$. Let
$\boldsymbol\phi_1$ be the rejected transition in the first stage, for any
$\boldsymbol\phi_0, \boldsymbol\phi_1, \boldsymbol\phi_2 \in \Phi \backslash \{
\boldsymbol\phi_1^{(k-1)}, \ldots, \boldsymbol\phi_n^{(k-1)} \}$. It will be
proved that for such candidates the following holds:
\begin{align}
p_k(\boldsymbol\phi_0) f_2(\boldsymbol\phi_2 | \boldsymbol\phi_0) = p_k(\boldsymbol\phi_2)
f_2(\boldsymbol\phi_0 | \boldsymbol\phi_2). \label{eq:balance_cond}
\end{align}
As seen from the description in section \ref{sec:parallel_delayed} it follows that 
\begin{align}
f_k(\boldsymbol\phi_2 | \boldsymbol\phi_0 ) &=
\underbrace{\hat{p}_{k,n}(\boldsymbol\phi_1)}_{\text{generate }\boldsymbol\phi_1} \, \underbrace{(1
- a_1(\boldsymbol\phi_0, \boldsymbol\phi_1))}_{\text{reject }\boldsymbol\phi_1} \,
\underbrace{S_2(\boldsymbol\phi_2 | \boldsymbol\phi_0, \boldsymbol\phi_1)}_{\text{generate
}\boldsymbol\phi_2} \, \underbrace{a_2(\boldsymbol\phi_0, \boldsymbol\phi_2)}_{\text{accept
}\boldsymbol\phi_2}, 
\end{align}
where it is used the fact that AIMS-OPT generates first stage proposals with an independent
approximate distribution. Recall that the probability of a second stage proposal is 
\begin{align}
a_2(\boldsymbol\phi_0, \boldsymbol\phi_2) = \min \left\lbrace 1, \frac{p_k(\boldsymbol\phi_2) \,
S_2( \boldsymbol\phi_0 | \boldsymbol\phi_2, \boldsymbol\phi_1) \, (1-a_1(\boldsymbol\phi_2,
\boldsymbol\phi_1)) }{ p_k(\boldsymbol\phi_0) \, S_2( \boldsymbol\phi_2 | \boldsymbol\phi_0,
\boldsymbol\phi_1) \, (1-a_1(\boldsymbol\phi_0, \boldsymbol\phi_1)) } \right\rbrace
\end{align}
and the fact that for any two positive numbers $a,b$ the equality $a\, \min \{1,
b/a\} = b \, \min \{1, a/b \} $ is satisfied. With these two equalities we can
substitute the left hand side of equation \eqref{eq:balance_cond} as
\begin{align}
p_k(\boldsymbol\phi_0) f_2(\boldsymbol\phi_2 | \boldsymbol\phi_0) &=
\hat{p}_{k,n}(\boldsymbol\phi_1)\, \left[ p_k(\boldsymbol\phi_0)\, S_2( \boldsymbol\phi_2 |
\boldsymbol\phi_0, \boldsymbol\phi_1) \, (1-a_1(\boldsymbol\phi_0, \boldsymbol\phi_1)) \right] \,
a_2(\boldsymbol\phi_0, \boldsymbol\phi_2) \nonumber \\
&= \hat{p}_{k,n}(\boldsymbol\phi_1)\, \left[ p_k(\boldsymbol\phi_2)\, S_2( \boldsymbol\phi_0 |
\boldsymbol\phi_2, \boldsymbol\phi_1) \, (1-a_1(\boldsymbol\phi_2, \boldsymbol\phi_1)) \right] \,
a_2(\boldsymbol\phi_2, \boldsymbol\phi_0)\nonumber \\
&= p_k(\boldsymbol\phi_2) \, f_2(\boldsymbol\phi_0 | \boldsymbol\phi_2), 
\end{align}
which proves the detailed balance for the second stage proposal. Note that the
proof has been made with no further assumptions about the second stage proposal
distribution $S_2(\boldsymbol\phi_2|\boldsymbol\phi_0, \boldsymbol\phi_1)$, as
it can be defined from several candidates. In this work, a symmetric proposal
that ignores the rejected sample has been used since it can be interpreted as a
Random Walk safeguard against a possible ill approximation done by the
independent sampler.  

%---------------------------------------------------------------------------------------- 
%   REFERENCE LIST 
%----------------------------------------------------------------------------------------

%----------------------------------------------------------------------------------------

\end{document}